\documentclass{amsart}
\usepackage{mathrsfs}
\usepackage{amssymb}
\usepackage{amsmath}
\usepackage{amsfonts}
\usepackage{enumerate}
\usepackage{pifont}
\usepackage{enumerate}
\usepackage[pdftex]{graphicx}
\usepackage{verbatim}
\usepackage{amsthm}
\usepackage[all]{xy}
\usepackage{xypic}

\numberwithin{equation}{section}
\newtheorem{theorem}{Theorem}[section]

\newtheorem{definition}{Definition}[section]

\newtheorem{proposition}{Proposition}[section]
\newtheorem{remark}{Remark}[section]
\newtheorem{example}{Example}[section]

\newcommand{\8}{\infty}

\newcommand{\be}{\begin{eqnarray*}}
\newcommand{\ee}{\end{eqnarray*}}
\newcommand{\beq}{\begin{equation}}
\newcommand{\eeq}{\end{equation}}
\newcommand{\beqn}{\begin{equation*}}
\newcommand{\eeqn}{\end{equation*}}
\newcommand{\bs}{\begin{split}}
\newcommand{\es}{\end{split}}

\begin{document}

\title{A mathematical formalism of non-Hermitian quantum mechanics and observable-geometric phases}

\author{Zeqian Chen}

\address{Wuhan Institute of Physics and Mathematics, Innovation Academy for Precision Measurement Science and Technology, Chinese Academy of Sciences, 30 West District, Xiao-Hong-Shan, Wuhan 430071, China.}


\date{}
\maketitle
\markboth{Zeqian Chen}%
{Non-Hermitian quantum mechanics}

\begin{abstract}
We present a mathematical formalism of non-Hermitian quantum mechanics, following the Dirac-von Neumann formalism of quantum mechanics. In this formalism, the state postulate is the same as in the Dirac-von Neumann formalism, but the observable postulate should be changed to include para-Hermitian operators (spectral operators of scalar type with real spectrum) representing observable, as such both the measurement postulate and the evolution postulate must be modified accordingly. This is based on a Stone type theorem as proved here that the dynamics of non-Hermitian quantum systems is governed by para-unitary time evolution. The Born formula on the expectation of an observable at a certain state is given in the non-Hermitian setting, which is proved to be equal to the usual Born rule for every Hermitian observable, but for a non-Hermitian one it may depend on measurement via the choice of a metric operator associated with the non-Hermitian observable under measurement. Our formalism is nether Hamiltonian-dependent nor basis-dependent, but can recover both PT-symmetric and biorthogonal quantum mechanics, and it reduces to the Dirac-von Neumann formalism of quantum mechanics in the Hermitian setting. As application, we study observable-geometric phases for non-Hermitian quantum systems.
\end{abstract}

\maketitle

\tableofcontents

\section{Introduction}\label{Intro}

Non-Hermitian quantum theory regards as observable some non-Hermitian (not necessarily self-adjoint) operators, such as $PT$-symmetric, pseudo-Hermitian or biorthogonal quantum mechanics. In $PT$-symmetric quantum mechanics developed by Bender {\it et al.} \cite{BB1998, BBM1999}, the Hamiltonian $H$ is not necessarily Hermitian, but has the unbroken $PT$-symmetry so that all its spectra are real. Bender {\it et al.} \cite{BBJ2002} have shown that the Hamiltonian $H$ with the unbroken $PT$-symmetry is always Hermitian (or {\it self-adjoint} in mathematical texts) in a new inner product defined by a symmetric operator $C$ associated with $H.$ However, this physical Hilbert-space inner product is dependent on the Hamiltonian $H$ itself. Mostafazadeh \cite{Mosta2010} further developed pseudo-Hermitian quantum mechanics by employing the concept of pseudo Hermiticity first introduced by Dirac and Pauli \cite{Pauli1943}. On the other hand, in biorthogonal quantum mechanics as developed by Brody \cite{Brody2014}, the observables are determined by a previously chosen (unconditional) basis in the associated Hilbert space. In such a theory, since an unconditional basis is not necessarily orthogonal in a Hilbert space \cite{LT1977}, some observables are not represented by Hermitian operators, and meanwhile, some Hermitian operators are excluded from the observable when given a basis.

Namely, $PT$-symmetric or pseudo-Hermitian quantum mechanics is Hamiltonian-dependent in the sense that the physical Hilbert-space inner product is determined by the non-Hermitian Hamiltonian of a system, while biorthogonal quantum mechanics is basis-dependent in the sense that observables are determined by a chosen basis in the Hilbert space of a system. Note that the conventional quantum mechanics, that is the Dirac-von Neumann formalism \cite{Dirac1958, vN1955}, is nether a Hamiltonian-dependent nor basis-dependent theory. To this end, we need to present a mathematical formalism of non-Hermitian quantum mechanics, which is nether Hamiltonian-dependent nor basis-dependent, but can recover both $PT$-symmetric and biorthogonal quantum mechanics, and it reduces to the Dirac-von Neumann formalism in the Hermitian setting.

In our formalism, the (pure) states are represented by ray lines or nonzero vectors as the same as the usual, but the observable is represented by a spectral operator of scalar type with real spectrum ({\it para-Hermitian operators} in our notion, including all Hermitian operators), as such both the measurement postulate and the evolution postulate must be modified accordingly. This is based on a theorem of Antoine and Trapani \cite{AT2014}, stating that a densely defined closed operator $T$ is para-Hermitian if and only if there is a metric operator $G$ such that $G^\frac{1}{2}TG^{-\frac{1}{2}}$ is Hermitian (see Section \ref{Pre:para-Hermi}). In particular, a Stone type theorem associated with the para-unitary operators is proved that the dynamics of non-Hermitian quantum systems is governed by para-unitary time evolution (see Section \ref{Pre:Stone} for the details). There are more general concepts than para-Hermitian operators commonly found in the literature, such as {\it quasi-Hermitian} and {\it pseudo-Hermitian} operators (cf. \cite{AT2014, Mosta2010}). However, there seems not to exist a general theory of functional calculus for them (cf. \cite{DS1971}), yet a Stone type theorem cannot hold in general for such classes of operator. Mathematically, this explains the reason why we chose para-Hermitian operators representing observables in the formalism.

This formalism presents the non-Hermitian Born formula on the expectation of an observable at a certain state, which we give in terms of a metric operator associated with the observable under measurement. This formula is proved to be equal to the usual Born rule for any Hermitian observable, but for a non-Hermitian observable, it is usually dependent on the choice of a metric operator involved for the measurement of the associated non-Hermitian observable. Thus, in our formalism, a metric operator associated with a para-Hermitian operator plays a role of measurement only, as pointed out in \cite{SGH1992}, but it needs not to introduce a new inner product.

The paper is organized as follows. In Section \ref{Pre}, we include some definitions and preliminary results on para-Hermitian and para-unitary operators and evolution systems. In particular, we prove a Stone type theorem associated with the para-unitary operator. In Section \ref{Axiom}, we present a mathematical formalism of non-Hermitian quantum mechanics and give some examples for illustration. In Section \ref{PTBIqm}, we explain how to recover PT-symmetric and biorthogonal quantum mechanics from our formalism. As application, in Section \ref{GeoPhase}, we study the observable-geometric phase of a time-dependent non-Hermitian quantum system, which was introduced in \cite{Chen2020} for Hermitian quantum systems. We give a summary in Section \ref{Sum}. Finally, we include an appendix, namely Section \ref{App}, on the geometry of non-Hermitian observable space, which is needed for a geometrical description of the observable-geometric phase associated with a non-Hermitian quantum system.

\section{Preliminaries}\label{Pre}

In what follows, we utilize the standard notions and notations from functional analysis (cf. \cite{Rudin1991}). $\mathbb{C}$ denotes the complex field and $\mathbb{C}_* = \mathbb{C} \setminus \{0\}.$ We denote by $\mathbb{H}$ a complex separable Hilbert space with an inner product $\langle \cdot, \cdot \rangle,$ linear in the second entry, and $\mathbb{H}_* = \mathbb{H} \setminus \{0\}.$ By an {\it operator} $T$ in $\mathbb{H}$ we shall mean a linear mapping whose domain $\mathcal{D} (T)$ is a (not necessarily closed) linear subset of $\mathbb{H}$ and whose range $\mathcal{R} (T)$ lies in $\mathbb{H}.$ We always use $I$ denote the identity operator; $T^*$ the adjoint operator for any densely defined operator $T$ in $\mathbb{H};$ $\mathcal{B} (\mathbb{H})$ the algebra of all bounded operators; $\mathcal{O} (\mathbb{H})$ the set of all Hermitian (self-adjoint) operators; $\mathcal{P} (\mathbb{H})$ the set of all orthogonal projections; $\mathcal{U} (\mathbb{H})$ the group of all unitary operators on $\mathbb{H},$ and $\mathcal{T} (\mathbb{H})$ the group of all bounded operators with bounded inverse on $\mathbb{H}.$ Note that $\mathcal{P} (\mathbb{H}),$ $\mathcal{U} (\mathbb{H})$ and $\mathcal{T} (\mathbb{H})$ are all subsets of $\mathcal{B} (\mathbb{H}).$

\subsection{Spectral operators}\label{Pre:SpecOper}

If an operator $P \in \mathcal{B} (\mathbb{H})$ satisfies $P^2 = P,$ it is called a {\it projection} as in \cite{DS1971,RS1980I,Rudin1991} (or a {\it skew projection} in some literatures, cf.\cite{Ovch1993}), and is an {\it orthogonal projection} if in addition $P^* =P.$ For a projection $P,$ its adjoint operator $P^*$ and complementary operator $P^\bot = I -P$ are both projections. We denote by $\tilde{\mathcal{P}} (\mathbb{H})$ the set of all projections in $\mathbb{H},$ and thus $\tilde{\mathcal{P}} (\mathbb{H}) \supset \mathcal{P} (\mathbb{H}).$ For two commuting projections $P, Q \in \tilde{\mathcal{P}} (\mathbb{H}),$ the intersection $P\wedge Q$ is defined by
\be
P\wedge Q = P Q
\ee
with the range $P\wedge Q (\mathbb{H}) = P (\mathbb{H}) \cap Q (\mathbb{H}),$ and the union $P \vee Q$ by
\be
P \vee Q = P+Q - P Q,
\ee
with the range $P \vee Q (\mathbb{H}) = P (\mathbb{H}) + Q (\mathbb{H}) = \overline{\mathrm{span}} [P (\mathbb{H}) \cup Q (\mathbb{H})],$ the closed subspace of $\mathbb{H}$ spanned by the sets $P (\mathbb{H})$ and $Q (\mathbb{H}).$ The order $P \le Q$ between two commuting projections $P, Q \in \tilde{\mathcal{P}} (\mathbb{H})$ is defined to be $P (\mathbb{H}) \subset Q (\mathbb{H}).$ A Boolean algebra of projections in $\mathbb{H}$ is a subset of $\tilde{\mathcal{P}} (\mathbb{H})$ which is a Boolean algebra under operations $\wedge, \vee$ and $\le$ together with its zero and unit elements being the operators $0$ and $I$ in $\mathcal{B} (\mathbb{H})$ respectively.

\begin{definition}\label{df:SpecMeasure}{\rm (cf. \cite{Dunf1958, DS1971})}\;
Let $\Sigma$ be a $\sigma$-field of subsets of a non-empty set $\Omega.$ A spectral measure on $\Sigma$ is a map $\mathbf{E}$ from $\Sigma$ into a Boolean algebra of projections in $\mathbb{H}$ satisfying the following conditions:
\begin{enumerate}[$1)$]

\item $\mathbf{E} (\emptyset) =0$ and $\mathbf{E} (\Omega) =I.$

\item For any $A, B \in \Sigma,$ $\mathbf{E} (\Omega \setminus A) = \mathbf{E} (A)^\bot,$
\be
\mathbf{E} (A \cap B) = \mathbf{E} (A) \wedge \mathbf{E} (B), \quad \mathbf{E} (A \cup B) = \mathbf{E} (A) \vee \mathbf{E} (B).
\ee

\item $\mathbf{E} (A)$ is countably additive in $A$ in the strong operator topology, i.e., for every sequence $\{A_n\}$ of mutually disjoint sets in $\Sigma,$
\be
\mathbf{E} (\cup_n A_n) x = \sum_n \mathbf{E} (A_n)x
\ee
holds for any $x \in \mathbb{H},$ where the series of the right hand side converges in $\mathbb{H}$ in the norm topology.

\end{enumerate}
\end{definition}

\begin{remark}\label{rk:DualSpecMeasure}\rm
For a spectral measure $\mathbf{E}$ on $\Sigma,$ define $\mathbf{E}^* (A) = [\mathbf{E} (A)]^*$ for every $A \in \Sigma.$ Then $\mathbf{E}^*$ is also a spectral measure $\mathbf{E}$ on $\Sigma,$ called the dual of $\mathbf{E}.$
\end{remark}

Note that every spectral measure $\mathbf{E}$ on $\Sigma$ is bounded, i.e., $\sup_{A \in \Sigma} \| \mathbf{E} (A) \| < \8.$ In this case, the integral $\int_\Omega f(\omega) \mathbf{E} (d \omega)$ may be defined for every bounded $\Sigma$-measurable (complex-valued) function defined $\mathbf{E}$-almost everywhere on $\Omega.$ Recall that a function $f$ is defined $\mathbf{E}$-almost everywhere on $\Omega,$ if there exists $\Omega_0 \in \Sigma$ such that $\mathbf{E} (\Omega_0) = I$ and $f$ is well defined for every $\omega \in \Omega_0.$ It was shown (cf. \cite[X.1]{DS1963}) that this integral is a bounded homomorphism of the $C^*$-algebra of $\mathcal{B}(\Omega, \Sigma)$ of bounded $\Sigma$-measurable functions in $\Omega$ with the norm $\|f\| = \sup_{\omega \in \Omega} |f(\omega)|$ into the $C^*$-algebra $\mathcal{B} (\mathbb{H}),$ that is, for any $\alpha, \beta \in \mathbb{C}$ and for $f, g \in \mathcal{B}(\Omega, \Sigma),$
\be\begin{split}
\int_\Omega [\alpha f (\omega) + \beta g (\omega)] \mathbf{E} (d \omega) & = \alpha \int_\Omega f (\omega) \mathbf{E} (d \omega) + \beta \int_\Omega g (\omega) \mathbf{E} (d \omega),\\
\int_\Omega f (\omega) g (\omega) \mathbf{E} (d \omega) & = \int_\Omega f (\omega) \mathbf{E} (d \omega) \int_\Omega g (\omega) \mathbf{E} (d \omega),\\
\Big \| \int_\Omega f (\omega) \mathbf{E} (d \omega) \Big \| & \le C_\mathbf{E} \sup_{\omega \in \Omega} |f(\omega)|,
\end{split}\ee
where $C_\mathbf{E}$ is a positive constant depending only upon the spectral measure $\mathbf{E}.$

In the sequel, we will focus on the spectral measures on the $\sigma$-field of Borel sets in the complex plane $\mathbb{C},$ denoted by $\mathcal{B}_\mathbb{C}.$

\begin{definition}\label{df:SpecOp}{\rm (cf. \cite[Definition XVIII.2.1]{DS1971})}\;
A densely defined closed operator $T$ in $\mathbb{H}$ with the domain $\mathcal{D} (T)$ is called a spectral operator, if there is a spectral measure $\mathbf{E}$ on $\mathcal{B}_\mathbb{C}$ such that
\begin{enumerate}[$1)$]

\item $\mathbf{E}$ is regular, i.e., for any $x,y \in \mathbb{H},$ the complex-valued measure $A \mapsto \langle x, \mathbf{E} (A) y \rangle$ is regular on $\mathcal{B}_\mathbb{C},$

\item for any bounded set $A \in \mathcal{B}_\mathbb{C},$ $\mathbf{E} (A) \mathbb{H} \subset \mathcal{D} (T),$

\item for any $B \in \mathcal{B}_\mathbb{C},$ $\mathbf{E} (B) \mathcal{D} (T) \subset \mathcal{D} (T)$ and
\be
T \mathbf{E} (B) x = \mathbf{E} (B) T x
\ee
for all $x \in \mathcal{D} (T),$

\item for any $B \in \mathcal{B}_\mathbb{C},$ the spectral set $\sigma (T|\mathbf{E} (B) (\mathbb{H}))$ of the restriction $T|\mathbf{E} (B) (\mathbb{H})$ of $T$ to $\mathbf{E} (B) (\mathbb{H})$ is contained in the closure $\bar{B}$ of $B,$ i.e.,
\be
\sigma (T|\mathbf{E} (B) (\mathbb{H})) \subset \bar{B},
\ee
where the domain $\mathcal{D} (T|\mathbf{E} (B) (\mathbb{H})) = \mathcal{D} (T) \cap \mathbf{E} (B) (\mathbb{H}).$

\end{enumerate}
The spectral measure $\mathbf{E}$ is called the {\it spectral resolution} or {\it resolution of the identity} for $T.$
\end{definition}

Note that the spectral measure $\mathbf{E}$ is uniquely determined by $T,$ i.e., the spectral resolution of a densely defined closed operator in $\mathbb{H}$ is unique whenever it exists (cf. \cite[Theorem XVIII.2.5]{DS1971}).

\begin{definition}\label{df:SpecOpScalar}{\rm (cf. \cite[Definition XVIII.2.12]{DS1971})}\;
A densely defined closed operator $T$ in $\mathbb{H}$ with the domain $\mathcal{D} (T)$ is of scalar type, if there is a spectral measure $\mathbf{E}$ on $\mathcal{B}_\mathbb{C}$ such that
\be
\mathcal{D} (T) = \{ x \in \mathbb{H}: \lim_n T_n x\;\text{exists}\}
\ee
and
\be
T x = \lim_n  T_n x,\quad \forall x \in \mathcal{D} (T),
\ee
where
\be
T_n = \int_{\{z \in\mathbb{C}: |z| \le n\} } z \mathbf{E} (d z).
\ee
The spectral measure $\mathbf{E}$ is said to be the spectral resolution for $T.$
\end{definition}

\begin{remark}\label{rk:ScalarTypeOp}\rm
It is shown in \cite[Lemma XVIII.2.13]{DS1971} that a scalar type operator $T$ in the sense of Definition \ref{df:SpecOpScalar} is a spectral operator in the sense of Definition \ref{df:SpecOp} and the spectral resolution of $T$ is unique. Thus, a scalar type operator is also called a {\it spectral operator of scalar type}. A theorem of Wermer (cf. \cite[Theorem XV.6.2.4]{DS1971}) states that a bounded spectral operator $T$ of scalar type is equivalent to a normal operator, that is, there exists a bounded self-adjoint operator $K$ with bounded inverse $K^{-1}$ such that the operator $K T K^{-1}$ is a normal operator.
\end{remark}

\subsection{Para-Hermitian operators}\label{Pre:para-Hermi}

Now we are ready to introduce the notion of {\it para-Hermitian} operators, which plays an essential role in the mathematical formulation of non-Hermitian quantum mechanics.

\begin{definition}\label{df:paraHermiOp}
A densely defined closed operator $T$ in $\mathbb{H}$ is called a para-Hermitian operator, if it is a spectral operator of scalar type with real spectrum, namely $\sigma (T) \subset \mathbb{R}.$

\end{definition}

We denote by $\tilde{\mathcal{O}} (\mathbb{H})$ the set of all para-Hermitian operators. Thus, $\mathcal{O} (\mathbb{H}) \subset \tilde{\mathcal{O}} (\mathbb{H}).$

Recall that an operator $G$ in $\mathbb{H}$ is called a {\it metric operator} (cf. \cite{AT2014, Mosta2010}), if $G$ is a bounded and strictly positive self-adjoint operator having bounded inverse $G^{-1}.$ Given a metric operator $G$ in $\mathbb{H},$ we can define a new inner product $\langle \cdot, \cdot \rangle_G$ in $\mathbb{H}$ by $\langle u, v \rangle_G = \langle u, Gv \rangle$ for any $u,v \in \mathbb{H}.$ Then the induced norm $\|u\|_G = \|G^\frac{1}{2} u \|$ is really equivalent to the original norm of $\mathbb{H}.$

\begin{proposition}\label{prop:paraHop}{\rm (cf. \cite[Proposition 3.12]{AT2014})}
Let $T$ be a densely defined closed operator in $\mathbb{H}.$ Then the following statements are
equivalent:
\begin{enumerate}[\rm 1)]

\item $T$ is a para-Hermitian operator.

\item There exists a metric operator $G$ such that $T$ is self-adjoint with respect to the inner $\langle \cdot, \cdot \rangle_G.$

\item There exists a metric operator $G$ such that $G^\frac{1}{2} T G^{-\frac{1}{2}}$ is self-adjoint.

\end{enumerate}
\end{proposition}

\begin{remark}\label{rk:MetricOp}\rm
The metric operator $G$ associated with a para-Hermitian operator $T$ in the above proposition is dependent on $T$ itself, and needs not to be unique in general (see Example \ref{ex:NHObM} below). We denote by $\mathcal{M} (T)$ the set of all metric operator $G$ associated with a para-Hermitian operator $T.$ Evidently, for any Hermitian operator $T$ the identity operator $I \in \mathcal{M} (T)$ but $I$ is not necessarily a unique metric operator associated with a Hermitian operator.
\end{remark}

There are more general concepts than para-Hermitian operators commonly found in the literature, such as {\it quasi-Hermitian} and {\it pseudo-Hermitian} operators (cf. \cite{AT2014, Mosta2010}), we include their definitions here for the sake of convenience.

\begin{definition}\label{df:quasipseudoHermiOp}
Let $T$ be a densely defined closed operator in $\mathbb{H}.$
\begin{enumerate}[\rm 1)]

\item $T$ is called a quasi-Hermitian operator, if there exists a bounded and strictly positive self-adjoint operator $G$ such that
\be
G T = T^* G.
\ee

\item $T$ is called a pseudo-Hermitian operator, if there exists a bounded self-adjoint operator $\eta$ with bounded inverse $\eta^{-1},$ such that
\be
T^* = \eta T \eta^{-1}.
\ee

\end{enumerate}
\end{definition}

\begin{remark}\label{rk:QuasiPseudoOp}\rm
By definition, a quasi-Hermitian operator is para-Hermitian if the operator $G$ has bounded inverse $G^{-1},$ while a pseudo-Hermitian operator is para-Hermitian if the operator $\eta$ is a positive operator. Note that, the definitions of quasi-Hermitian and pseudo-Hermitian operators have been respectively adapted to the cases of a unbounded metric operator $G$ (cf. \cite{AT2014}) and a unbounded self-adjoint operator $\eta$ (cf. \cite{Mosta2013}).
\end{remark}


\begin{definition}\label{df:FunctCalculusparaHop}{\rm (cf. \cite[Definition XVIII.2.10]{DS1971})}\;
Let $T$ be a spectral operator of scalar type with the spectral resolution $\mathbf{E}$ on $\mathcal{B}_\mathbb{C}.$  For any $\mathcal{B}_\mathbb{C}$-measurable function $f,$ we define $f (T)$ by
\be
f(T) x = \lim_n T(f_n) x,\quad \forall x \in \mathcal{D} (f(T)),
\ee
where
\be\begin{split}
\mathcal{D} (f(T)) =& \{x \in \mathbb{H}: \lim_n T(f_n) x\;\text{exists}\},\\
T(f_n) = & \int_\mathbb{C} f_n (z) \mathbb{E} (d z),
\end{split}\ee
and
\be
f_n (z) = \left \{\begin{split} & f(z),\quad |f(z)| \le n,\\
& 0, \quad |f(z)| >0.
\end{split}\right.
\ee
\end{definition}

\begin{remark}\label{rk:FunctCalculusparaHop}\rm
It is shown in \cite[Theorem XVIII.2.17]{DS1971} that $f(T)$ in the above definition is a spectral operator of scalar type with the spectral resolution $\mathbf{E}_f (E) = \mathbf{E} (f^{-1} (E))$ for any $E \in \mathcal{B}_\mathbb{C}.$
\end{remark}

Thus, we have the well-defined functional calculus for para-Hermitian operators, which plays a role in the dynamics of non-Hermitian quantum mechanics as called the Stone-type theorem in the sequel. However, there seems no such functional calculus for either {\it quasi-Hermitian} or {\it pseudo-Hermitian} operators. Mathematically, this is the reason why we use para-Hermitian operators representing the observable beyond quasi-Hermitian and pseudo-Hermitian operators.

\subsection{A Stone-type theorem}\label{Pre:Stone}

At first, we need to introduce the notion of a {\it para-unitary operator}, corresponding to the one of a para-Hermitian operator.

\begin{definition}\label{df:paraUOp}
A bounded spectral operator $U$ of scalar type is said to be para-unitary if $\sigma (U) \subset \mathbb{T},$ namely $|\lambda| =1$ for all $\lambda \in \sigma (U).$
\end{definition}

We denote by $\tilde{\mathcal{U}} (\mathbb{H})$ the set of all para-unitary operators in $\mathbb{H}.$ Thus, $\mathcal{U} (\mathbb{H}) \subset \tilde{\mathcal{U}} (\mathbb{H}).$

\begin{proposition}\label{prop:paraUop}
A bounded spectral operator $U$ of scalar type in $\mathbb{H}$ is para-unitary if and only if there exists a metric operator $G$ such that $G^\frac{1}{2} U G^{-\frac{1}{2}}$ is unitary.
\end{proposition}

\begin{proof}
Suppose that $U$ is a para-unitary operator. By a theorem of Wermer (cf. \cite[Theorem XV.6.4]{DS1971}), there exists a bounded self-adjoint operator $K$ with bounded inverse $K^{-1}$ such that $K U K^{-1}$ is normal. Since $\sigma (U) \subset \mathbb{T},$ then $\sigma (K U K^{-1}) \subset \mathbb{T}$ and so $K U K^{-1}$ is unitary with the spectral decomposition
\be
K U K^{-1} = \int_\mathbb{T} \lambda \mathbf{E} (d \lambda),
\ee
where $\mathbf{E}$ is a self-adjoint spectral resolution. Putting $G = |K|^2,$ by the polar decomposition we have $K = V G^\frac{1}{2}$ with $V$ unitary such that
\be
G^\frac{1}{2} U G^{-\frac{1}{2}} = \int_\mathbb{T} \lambda \mathbf{F} (d \lambda),
\ee
where $\mathbf{F} (\cdot) = V^{-1}\mathbf{E} (\cdot) V$ is a self-adjoint spectral resolution. Thus, $G^\frac{1}{2} U G^{-\frac{1}{2}}$ is unitary.

Conversely, if there exists a bounded and strictly positive self-adjoint operator $G$ with bounded inverse $G^{-1},$ such that $G^\frac{1}{2} U G^{-\frac{1}{2}}$ is unitary, then
\be
U = \int_\mathbb{T} \lambda G^{-\frac{1}{2}} \mathbf{E} G^\frac{1}{2}(d \lambda)
\ee
where $\mathbf{E}$ is a self-adjoint spectral resolution. Clearly, $\mathbf{F} (\cdot) = G^{-\frac{1}{2}}\mathbf{E} (\cdot) G^\frac{1}{2}$ is a spectral resolution for $U$ and $\sigma (U) \subset \mathbb{T}.$ Hence, $U$ is a para-unitary operator. This completes the proof.
\end{proof}

The following proposition shows the relationship between para-Hermitian and para-unitary operators through function calculus about spectral operators of scalar type.

\begin{proposition}\label{prop:paraHUop}
Let $H$ be a para-Hermitian operator. If $f(z) = e^{\mathrm{i}z},$ then $f(H)$ is a para-unitary operator, denoted by $e^{\mathrm{i} H}.$ 
\end{proposition}

\begin{proof}
Let $H$ be a para-Hermitian operator and $f(z) = e^{\mathrm{i}z}.$ By \cite[Theorem XVIII.2.21]{DS1971}, $\sigma (f(H)) = \overline{f(\sigma(H))} \subset \mathbb{T}.$ Therefore, by \cite[Theorem XVIII.2.11(c) and Theorem XVIII.2.17]{DS1971}, $e^{\mathrm{i}H}$ is a para-unitary operator.
\end{proof}

In what follows, we will prove Stone's theorem for the one-parameter group of para-unitary operators, which is fundamental for the dynamics of non-Hermitian quantum mechanics.

\begin{proposition}\label{prop:paraUgro}
Let $H$ be a para-Hermitian operator and define $U(t) = e^{\mathrm{i} t H}$ for every $t \in \mathbb{R}.$ Then
\begin{enumerate}[\rm 1)]

\item For each $t \in \mathbb{R},$ $U(t)$ is a para-unitary operator and
\be
U(t+s) = U(t) U(s),\quad \forall t,s \in \mathbb{R}.
\ee

\item For each $x \in \mathbb{H},$ $\lim_{t \to t_0} U(t) x = U(t_0) x$ in $\mathbb{H},$ that is, $t \mapsto U(t)$ is strongly continuous.

\item $U(0) = I$ and $\{U(t): t \in \mathbb{R} \}$ is a bounded Abelian group of para-unitary operators such that $U(t)^{-1} = U(-t)$ for every $t \in \mathbb{R}.$

\item For every $x \in \mathcal{D} (H),$
\be
\lim_{t \to 0} \frac{U(t) x - x}{t} = \mathrm{i} H x.
\ee

\item If $x \in \mathbb{H}$ such that $\lim_{t \to 0} (U(t) x - x)/t$ exists, then $x \in \mathcal{D} (H).$

\end{enumerate}
\end{proposition}

\begin{proof}
1) follows immediately from Proposition \ref{prop:paraHUop} and the functional calculus for the complex-valued function $e^{\mathrm{i} t z}.$ To prove 2) note that
\be
\|e^{\mathrm{i} t H} x - x \|^2 = \int_\mathbb{R} \int_\mathbb{R} (e^{\mathrm{i} t \lambda} -1) (e^{\mathrm{i} t \gamma} -1) \langle \mathbf{E}^* (d \gamma) \mathbf{E} (d \lambda)x, x \rangle,
\ee
where $\mathbf{E}$ is the spectral resolution of $H$ and $\mathbf{E}^*$ is the dual of $\mathbf{E}$ (cf. Remark \ref{rk:DualSpecMeasure}). Define $\mathbf{F} (A \times B) = \mathbf{E}^* (A) \mathbf{E} (B)$ for any $A,B \in \mathcal{B}_\mathbb{R}$ (the $\sigma$-algebra of Borel sets in $\mathbb{R}$). Then $\mathbf{F}$ extends to a bounded operator-valued measure on $\mathcal{B}_{\mathbb{R}^2}$ (the $\sigma$-algebra of Borel sets in $\mathbb{R}^2$) such that
\be
\|e^{\mathrm{i} t H} x - x \|^2 = \int_{\mathbb{R}^2} (e^{\mathrm{i} t \lambda} -1) (e^{\mathrm{i} t \gamma} -1) \langle \mathbf{F} (d \gamma \times d \lambda)x, x \rangle.
\ee
Note that $\mu_x (K) = \langle \mathbf{F} (K)x, x \rangle$ is a complex measure on $\mathcal{B}_{\mathbb{R}^2}.$ Since $|(e^{\mathrm{i} t \lambda} -1) (e^{\mathrm{i} t \gamma} -1)| \le 4$ and
\be
\|e^{\mathrm{i} t H} x - x \|^2 \le \int_{\mathbb{R}^2} |(e^{\mathrm{i} t \lambda} -1) (e^{\mathrm{i} t \gamma} -1)| |\mu_x| (d \gamma \times d \lambda),
\ee
we conclude that $\lim_{t \to 0} \|e^{\mathrm{i} t H} x - x \|^2 =0$ by the Lebesgue dominated convergence theorem. Thus $t \mapsto U(t)$ is strongly continuous at $t=0,$ which implies by the group property that $t \mapsto U(t)$ is strongly continuous at any $t \in \mathbb{R}.$

For 3), by Proposition \ref{prop:paraHop} there is a metric operator $G$ such that $G^\frac{1}{2} H G^{-\frac{1}{2}}$ is self-adjoint. By the uniqueness of spectral resolution and functional calculus, we conclude that $G^\frac{1}{2} U(t) G^{-\frac{1}{2}} = e^{\mathrm{i} t G^\frac{1}{2} H G^{-\frac{1}{2}}}$ is unitary for every $t \in \mathbb{R}.$ Thus,
\be
\| U(t) \| = \| G^{-\frac{1}{2}} e^{\mathrm{i} t G^\frac{1}{2} H G^{-\frac{1}{2}}} G^\frac{1}{2} \| \le \| G^{-\frac{1}{2}} \| \| G^\frac{1}{2} \|
\ee
for all $t \in \mathbb{R},$ i.e., $\{U(t): t \in \mathbb{R} \}$ is a bounded set of operators.

For 4) and 5), as above there is a metric operator $G$ such that $G^\frac{1}{2} U(t) G^{-\frac{1}{2}} = e^{\mathrm{i} t G^\frac{1}{2} H G^{-\frac{1}{2}}}$ is unitary for every $t \in \mathbb{R}.$ Thus, by \cite[Theorem VIII.7 (c) and (d)]{RS1980I} we obtain 4) and 5), since $x \in \mathcal{D} (H)$ if and only if $G^\frac{1}{2} x \in \mathcal{D} (G^\frac{1}{2} H G^{-\frac{1}{2}}).$
\end{proof}

\begin{remark}\label{rk:SkewUnitaryGroup}\rm
An operator-valued function $t \mapsto U(t)$ from $\mathbb{R}$ into $\tilde{\mathcal{U}} (\mathbb{H})$ satisfying $1)$ and $2)$ is called {\it a strongly continuous one-parameter para-unitary group}.
\end{remark}

The following theorem says that every strongly continuous bounded one-parameter para-unitary group arises as the exponential of a para-Hermitian operator, that is, Stone's theorem holds true in the para-Hermitian case.

\begin{theorem}\label{th:StoneTh}
Let $(U(t):\; t \in \mathbb{R})$ be a strongly continuous bounded one-parameter para-unitary group on a Hilbert space $\mathbb{H}.$ Then there is a para-Hermitian operator $H$ on $\mathbb{H}$ so that $U(t) = e^{\mathrm{i}t H}$ for every $t \in \mathbb{R}.$
\end{theorem}

\begin{proof}
By \cite[Lemma XV.6.1]{DS1971}, there exists a bounded self-adjoint operator $K$ with bounded inverse $K^{-1}$ such that $K U(t) K^{-1}$ is unitary for every $t \in \mathbb{R}.$ By the polar decomposition of $K,$ $G^\frac{1}{2} U(t) G^{-\frac{1}{2}}$ is unitary for every $t \in \mathbb{R},$ where $G = |K|^2$ is a bounded and strictly positive self-adjoint operator with bounded inverse $G^{-1}.$ Since $(G^\frac{1}{2} U(t) G^{-\frac{1}{2}}: t \in \mathbb{R})$ is a strongly continuous one-parameter unitary group, by Stone's theorem (cf.\cite[Theorem VIII.8]{RS1980I}) there exists a self-adjoint operator $A$ such that $e^{\mathrm{i} t A} = G^\frac{1}{2} U(t) G^{-\frac{1}{2}}$ for every $t \in \mathbb{R}.$ By the uniqueness of spectral resolution and functional calculus, we conclude that
\be
U(t) = G^{-\frac{1}{2}} e^{\mathrm{i} t A} G^\frac{1}{2} = e^{\mathrm{i} t H},\quad \forall t \in \mathbb{R},
\ee
where $H = G^{-\frac{1}{2}} A  G^\frac{1}{2}$ is a para-Hermitian operator by Proposition \ref{prop:paraHop}.
\end{proof}

\begin{remark}\label{rk:paraUnitaryGroupInfGenerator}\rm
If $(U(t):\; t \in \mathbb{R})$ is a strongly continuous bounded one-parameter para-unitary group, then the para-Hermitian operator $H$ with $U(t)= e^{\mathrm{i}t H}$ ($\forall t \in \mathbb{R}$) is called the {\it infinitesimal generator} of $(U(t):\; t \in \mathbb{R}).$
\end{remark}

\begin{remark}\label{rk:paraUnitaryGroupPl}\rm
Concerning Theorem \ref{th:StoneTh}, a question arises: Whether does there exist a strongly continuous one-parameter para-unitary group on an infinite-dimensional Hilbert space $\mathbb{H}$ which is unbounded as a subset of $\mathcal{B} (\mathbb{H})$? At the time of this writing, we have no such example.
\end{remark}

\subsection{Evolution systems}\label{Pre:EvoSys}

Any two-parameter family of bounded operators $\{U(t,s) \in \mathcal{B} (\mathbb{H}): s,t \in [0, T]\}$ is said to be an {\it evolution system} (cf. \cite{SG2017}), if it satisfies the following conditions:
\begin{enumerate}[(i)]

\item $U(t,t) = I$ and $U(t, r) U (r, s) = U(t,s)$ for all $s,r,t \in [0, T];$

\item $(t,s) \mapsto U(t,s)$ is strongly continuous on $[0,T] \times [0, T].$

\end{enumerate}
Note that by (i), $U(t,s)$ are all bounded operators with bounded inverse and $U(t,s)^{-1} = U(s,t),$ namely $U(t,s) \in \mathcal{T} (\mathbb{H}).$ In some literatures, an evolution system is only assumed to satisfy (i) and (ii) on the triangle region $0\le s \le t \le T$ (cf. \cite{Pazy1983}). In this case, $U(t,s)$ need not to be invertible.

Let $\{A(t): t \in [0, T]\}$ be a family of densely defined closed operators in $\mathbb{H}$ with a property that there exists a dense subset $\mathbb{D}$ of $\mathbb{H}$ such that $\mathbb{D} \subset \mathcal{D} (A(t))$ for all $t \in [0,T].$ If a evolution system $\{U(t,s) \in \mathcal{T} (\mathbb{H}): s,t \in [0, T]\}$ satisfies the condition that $U(t,s) \mathbb{D} \subset \mathbb{D}$ for all $s,t \in[0,T],$ and for any $v \in \mathbb{D}$ and $s \in [0,T],$ the map $t \mapsto U(t,s) v$ is continuously differentiable in $[0, T]$ such that
\begin{equation}\label{eq:EvoSysEqu}
\frac{d}{d t} U(t,s)v = A(t) U(t,s)v, \quad \forall t \in [0,T],
\end{equation}
then it is called an evolution system for $A(t)$ on $\mathbb{D}$ (cf. \cite{SG2017}).

It is well known that if $A(t)$'s are all bounded operators such that $[0,T] \ni t \mapsto A(t) \in \mathcal{B}(\mathbb{H})$ is strongly continuous, then there exists a unique evolution system $\{U(t,s):t,s \in [0,T]\}$ such that the evolution equations
\begin{equation}\label{eq:EvoSysEquI}
\frac{d}{d t} U(t,s) = A(t) U(t,s),
\end{equation}
and
\begin{equation}\label{eq:EvoSysEquII}
\frac{d}{d t} U(s,t) = -U(s,t) A(t)
\end{equation}
hold in the strong topology of $\mathcal{B} (\mathbb{H})$ for any $s \in [0, T]$ (see \cite{Pazy1983} for the details).

We refer to \cite{NZ2009, Schmid2016, SG2017} for the details on the existence of the evolution systems for a family of densely defined closed operators in a Hilbert space. In fact, by \cite[Theorem 2.1]{SG2017}, we have the following result:

\begin{proposition}\label{prop:paraEvoSys}
Let $\{h(t): t \in [0, T]\}$ be a family of para-Hermitian operators in $\mathbb{H},$ having the same domain $\mathbb{D},$ namely $\mathbb{D} = \mathcal{D} (h(t))$ for all $t \in [0,T].$ If there exists $\omega >0$ such that
\be
\| e^{-\mathrm{i} s h(t)} \| \le e^{\omega |s|},\quad \forall s \in \mathbb{R},\forall t \in [0,T],
\ee
and if the map $[0,T] \ni t \mapsto h(t) \in \mathcal{B} (\mathbb{D}, \mathbb{H})$ is continuous and of bounded variation, where $\mathbb{D}$ is endowed with the graph norm of $h(0),$ then there exists a unique evolution system $\{U(t,s):t,s \in [0,T]\}$ for $A(t)= - \mathrm{i} h(t)$ on $\mathbb{D}.$
\end{proposition}

\begin{proof}
This is so, because every $A(t) = - \mathrm{i} h(t)$ generates a strongly continuous group $\{e^{-\mathrm{i} s h(t)}: s \in \mathbb{R}\}$ by Proposition \ref{prop:paraUgro}.
\end{proof}

\section{Mathematical axiom}\label{Axiom}

Following the Dirac-von Neumann formalism of quantum mechanics \cite{Dirac1958, vN1955}, we present a mathematical formalism of non-Hermitian quantum mechanics in what follows. Precisely, this formalism includes the following five postulates:

\begin{definition}\label{df:MathFNHQM}
The mathematical formalism of non-Hermitian quantum mechanics is defined by a set of postulates as follows:
\begin{enumerate}[$(P_1)$]

\item {\bf The state postulate}\; Associated with a non-Hermitian quantum system is a complex separable Hilbert space $\mathbb{H},$ the system at any given time is described by a state, which is determined by a nonzero vector in $\mathbb{H}.$

\item {\bf The observable postulate}\; Each observable for a non-Hermitian quantum system associated with a complex separable Hilbert space $\mathbb{H}$ is represented by a para-Hermitian operator in $\mathbb{H}.$

\item {\bf The measurement postulate}\; For an observable represented by a para-Hermitian operator $A$ in $\mathbb{H},$ if $G$ is a metric operator associated with $A,$ then $G$ introduces a measurement context for the observable $A$ such that the expectation of $A$ at a certain state determined by a nonzero vector $\psi$ with $ G^{-\frac{1}{2}} \psi \in \mathcal{D} (A)$ is given by
\beq\label{eq:BornRuleNH}
\langle A \rangle_{\psi, G} = \frac{\langle \psi, G^\frac{1}{2} A G^{-\frac{1}{2}} \psi \rangle}{\| \psi \|^2}.
\eeq
In particular, if $A$ has a discrete spetrum $\{\lambda_n\}_{n \ge 1},$ whose eigenstates $\{e_n\}_{n \ge 1}$ is a unconditional basis in $\mathbb{H},$ then the expectation of $A$ at $\psi$ with $G^{-\frac{1}{2}} \psi \in \mathcal{D} (A),$ where $G = \sum_{n \ge 1} |e^*_n\rangle \langle e^*_n|$ being a metric operator for $A,$ is given by
\beq\label{eq:BornRuleNHV}
\langle A \rangle_{\psi, \Pi} = \sum^\8_{n = 1} \lambda_n \frac{|\langle e^*_n, G^{-\frac{1}{2}} \psi\rangle |^2}{\| \psi \|^2},
\eeq
under the measurement $\Pi = \{|e_n\rangle \langle e^*_n|: n \ge 1 \}$ or equivalently in the measurement context of $G.$ In this case, $\psi$ will be changed to the state $e_n$ with probability
\beq\label{eq:TransProba}
p(\psi | e_n) = \frac{|\langle e^*_n, G^{-\frac{1}{2}} \psi\rangle |^2}{\| \psi \|^2}
\eeq
for each $n \ge 1.$

\item {\bf The evolution postulate}\; The system described by vectors is changed with time according to the Schr\"{o}dinger equation
\beq\label{eq:BiSchrEqu}
\mathrm{i} \frac{d \psi (t) }{d t} = H \psi (t),
\eeq
where $H$ is a para-Hermitian operator, which is called the energy operator of the system.

\item {\bf The composite-systems postulate}\; The Hilbert space associated with a composite non-Hermitian quantum system is the Hilbert space tensor product of the Hilbert spaces of its components. If systems numbered $1$ through $n$ are prepared in states $\psi_k,$ $k=1,\ldots, n,$ then the joint state of the composite total system is the tensor product $\psi_1 \otimes \cdots \otimes \psi_n.$

\end{enumerate}
\end{definition}

\begin{remark}\label{rk:NHQM}\rm
\begin{enumerate}[$1)$]

\item Both postulates $(P_1)$ and $(P_5)$ are the same as in the Dirac-von Neumann formalism of quantum mechanics. For the sake of completeness, we include them here.

\item Let $A$ be a para-Hermitian such that $A = \sum_{n \ge 1} \lambda_n |e_n\rangle \langle e^*_n|$ with a discrete spectrum whose eigenstates $(e_n)_{n \ge 1}$ constitute a unconditional basis in $\mathbb{H}.$ Define $G = \sum_{n \ge 1} |e^*_n\rangle \langle e^*_n|.$ Then $e^*_n = G e_n$ for every $n \ge 1,$ and $(e_n)_{n \ge 1}$ is orthogonal in the inner product $\langle \phi, \psi\rangle_G = \langle \phi, G \psi\rangle$ defined by $G.$ Moreover, we have
\be
\| \psi \|^2 = \langle G^{-\frac{1}{2}}\psi, G^{-\frac{1}{2}}\psi\rangle_G = \sum_{n \ge 1} |\langle e^*_n, G^{-\frac{1}{2}} \psi\rangle|^2,
\ee
and
\be
\sum_n p(\psi | e_n) = 1.
\ee
Thus, the formula \eqref{eq:BornRuleNHV} of the discrete case coincides with \eqref{eq:BornRuleNH}.

\item By the non-Hermitian Born formula \eqref{eq:BornRuleNH}, each state is scalar free and uniquely determined by a complex line through the origin of $\mathbb{H},$ i.e.,
\beq\label{eq:BornRuleScalarfree}
\langle A \rangle_{\alpha \psi, G} = \langle A \rangle_{\psi, G}
\eeq
for any nonzero scalar $\alpha \in \mathbb{C}.$ In what follows, without specified otherwise, we always use a unit vector to represent a state, simply called a {\it vector state}.
\end{enumerate}

\end{remark}

\begin{proposition}\label{prop:BornRuleHermit}\rm
Let $A$ be a Hermitian operator in $\mathbb{H}.$ If $G$ is a metric operator associated with $A,$ then
\beq\label{eq:BornRuleHermit}
\langle A \rangle_{\psi, G} = \langle A \rangle_\psi : = \frac{\langle \psi, A \psi \rangle}{\| \psi \|^2}
\eeq
for any nonzero $\psi \in \mathbb{H}.$
\end{proposition}

\begin{proof}
Since $G$ and $G^\frac{1}{2} A G^{-\frac{1}{2}}$ are both self-adjoint, it follows that
\be
G^\frac{1}{2} A G^{-\frac{1}{2}} = G^{-\frac{1}{2}} A G^\frac{1}{2}
\ee
and so $GA = AG,$ i.e., $G$ commutes with $A.$ Note that $\sigma (G)$ is a bounded closed set in $(0, \8)$ by the assumption. Thus $f(x) = x^\frac{1}{2}$ is a Borel function in $\sigma (G).$ By the spectral theorem (cf. \cite[Theorem 13.33]{Rudin1991}) and functional calculus, we conclude that $G^\frac{1}{2} A = A G^\frac{1}{2}.$ This implies the required \eqref{eq:BornRuleHermit}.
\end{proof}

\begin{remark}\label{rk:NHQMtoDvN}\rm
By Proposition \ref{prop:BornRuleHermit}, the non-Hermitian Born formula \eqref{eq:BornRuleNH} is independent of the choice of a metric operator $G$ associated with a Hermitian operator $A.$ Therefore, the mathematical formalism of non-Hermitian quantum mechanics as in Definition \ref{df:MathFNHQM} is an extension of the Dirac-von Neumann formalism of quantum mechanics to the non-Hermitian setting.
\end{remark}

The following example shows that the non-Hermitian Born formula \eqref{eq:BornRuleNH} is not equal to the usual Born rule for a non-Hermitian observable in general.

\begin{example}\rm
Consider the operator $A: \mathbb{C}^2 \mapsto \mathbb{C}^2$ defined by the matrix
\be
A = \left (
\begin{matrix}
0 & 1 \\
4 & 0
\end{matrix} \right )
\ee
in the standard basis of $\mathbb{C}^2$ and is not Hermitian. Define $G: \mathbb{C}^2 \mapsto \mathbb{C}^2$ by the matrix
\be
G = \left (
\begin{matrix}
1 & 0 \\
0 & \frac{1}{4}
\end{matrix} \right ).
\ee
Then $G^\frac{1}{2} A G^{- \frac{1}{2}} = 2 \sigma_x $ is Hermitian and so $A$ is para-Hermitian in $\mathbb{C}^2,$ that is, $A$ is a non-Hermitian observable. For $\psi = \frac{1}{\sqrt{2}} \left ( \begin{matrix} 1 \\  - \mathrm{i} \end{matrix} \right ),$ we have
\be
\langle A \rangle_\psi = \frac{3}{2} \mathrm{i}, \quad \langle A \rangle_{\psi, G} = 0,
\ee
and thus $\langle A \rangle_{\psi, G} \not= \langle A \rangle_\psi.$
\end{example}

The following example shows that the measurement of a non-Hermitian observable represented by a para-Hermitian operator may depend on the choice of a metric operator associated with it.

\begin{example}\label{ex:NHObM}\rm (cf. \cite[Section 3.5]{Mosta2010})\;
Consider the operator $A: \mathbb{C}^2 \mapsto \mathbb{C}^2$ defined by the matrix
\be
A = \frac{1}{2}\left (
\begin{matrix}
1+\delta & -1+\delta \\
1-\delta & -1 - \delta
\end{matrix} \right )
\ee
in the standard basis of $\mathbb{C}^2,$ where $\delta >0.$ It has two eigenvalues $\lambda_\pm = \pm \delta^\frac{1}{2}$ with the corresponding eigenstates
\be
e_+ (A) = c_+ \left (
\begin{matrix}
1 + \delta^\frac{1}{2}\\
1 - \delta^\frac{1}{2}
\end{matrix} \right ),\quad
e_- (A) = c_- \left (
\begin{matrix}
1 - \delta^\frac{1}{2}\\
1 + \delta^\frac{1}{2}
\end{matrix} \right ),
\ee
where $c_1, c_2 \in \mathbb{C}_*,$ and with the dual eigenstaes
\be
e^*_+ (A) = \frac{1}{4 \bar{c}_+} \left (
\begin{matrix}
1 + \delta^{-\frac{1}{2}}\\
1 - \delta^{-\frac{1}{2}}
\end{matrix} \right ),\quad
e^*_- (A) & = \frac{1}{4 \bar{c}_-} \left (
\begin{matrix}
1 - \delta^{-\frac{1}{2}}\\
1 + \delta^{-\frac{1}{2}}
\end{matrix} \right ),
\ee
such that
\be
A = \delta^\frac{1}{2} |e_+ (A) \rangle \langle e^*_+ (A)| - \delta^\frac{1}{2} |e_- (A)\rangle \langle e^*_- (A)|.
\ee
Hence, $A$ is para-Hermitian for all $\delta>0,$ and is Hermitian only when $\delta =1.$

Define $G: \mathbb{C}^2 \mapsto \mathbb{C}^2$ by the matrix
\be
G = r_+ \left (
\begin{matrix}
(1 + \delta^{-\frac{1}{2}})^2 & 1 - \delta^{-1} \\
1 - \delta^{-1} & (1 - \delta^{-\frac{1}{2}})^2
\end{matrix} \right ) + r_- \left (
\begin{matrix}
(1 - \delta^{-\frac{1}{2}})^2 & 1 - \delta^{-1} \\
1 - \delta^{-1} & (1 + \delta^{-\frac{1}{2}})^2
\end{matrix} \right ),
\ee
in the standard basis, where $r_+ = |4 c_+|^{-2}$ and $r_- = |4 c_-|^{-2}.$ By computation, one finds that $A$ is Hermitian in the inner product $(\cdot, \cdot)_G,$ and thus $G$ is a metric operator associated with $A.$ Note that $G$ has two positive eigenvalues
\be
\lambda_\pm = (1+ \delta^{-1}) (r_+ + r_-) \pm [(1+ \delta^{-1})^2 (r_+ + r_-)^2 - 4^2 \delta^{-1} r_+ r_-]^\frac{1}{2}.
\ee

Taking $\delta = \frac{1}{4},$ we have $\lambda_\pm = 5(r_+ + r_-) \pm \gamma$ with the corresponding eigenstates
\be\begin{split}
e_+ (G) = & \left (
\begin{matrix}
-3(r_+ + r_-)\\
\gamma -4(r_+ - r_-)
\end{matrix} \right ),\\
e_-(G) = & \left (
\begin{matrix}
-3(r_+ + r_-)\\
-\gamma -4(r_+ - r_-)
\end{matrix} \right ),
\end{split}\ee
where $\gamma = [25( r_+^2 + r_-^2) - 14 r_+ r_-]^\frac{1}{2}.$ It is then easy to see that the expectation $\langle A \rangle_{|0\rangle, G}$ of $A$ at $|0\rangle = \left (
\begin{matrix}
1\\
0
\end{matrix} \right )$ is dependent on the values of $r_+$ and $r_-,$ and thus $\langle A \rangle_{|0\rangle, G}$ depends on the choice of the metric operator $G$ associated with it.
\end{example}

\begin{example}\label{ex:para-PaulMat}\rm
Consider the non-Hermitian qubit system associated with the Hilbert space $\mathbb{C}^2.$ Given a real number $-\frac{\pi}{2} < \omega < \frac{\pi}{2},$ the non-Hermitian (deformed) Pauli matrices are defined by (cf. \cite{Brody2014})
\beq\label{eq:biPauliM}
\left \{
\begin{split}
\sigma^\omega_x & = \frac{1}{\cos \omega} \left (
\begin{matrix}
- \mathrm{i} \sin \omega & 1 \\
1 & \mathrm{i} \sin \omega
\end{matrix} \right ),\\
\sigma^\omega_y & = \left (
\begin{matrix}
0 & - \mathrm{i} \\
\mathrm{i} & 0
\end{matrix} \right ),\\
\sigma^\omega_z & = \frac{1}{\cos \omega} \left (
\begin{matrix}
1 & \mathrm{i}\sin \omega \\
\mathrm{i}\sin \omega & -1
\end{matrix} \right ).
\end{split}\right.\eeq
All $\sigma^\omega_x, \sigma^\omega_y,$ and $\sigma^\omega_z$ have eigenvalues $1$ and $-1,$ and satisfy the canonical commutation relations
\beq\label{eq:PauliCommuLation}
\sigma^\omega_x \sigma^\omega_y = \mathrm{i} \sigma^\omega_z,\; \sigma^\omega_y \sigma^\omega_z = \mathrm{i} \sigma^\omega_x,\;\sigma^\omega_z \sigma^\omega_x = \mathrm{i} \sigma^\omega_y.
\eeq

The eigenstates of $\sigma^\omega_x$ are
\beq\label{eq:Xeigenstate}
\left \{ \begin{split}
e_+ (\sigma^\omega_x) & = \frac{1}{\sqrt{2}} \left (
\begin{matrix}
1 \\
e^{\mathrm{i}\omega}
\end{matrix} \right ),\\
e_- (\sigma^\omega_x) & = \frac{1}{\sqrt{2}} \left (
\begin{matrix}
1 \\
- e^{-\mathrm{i}\omega}
\end{matrix} \right ),
\end{split}
\right.\eeq
and
\beq\label{eq:XstarEigenstate}
\left \{\begin{split}
e^*_+ (\sigma^\omega_x) & = \frac{1}{\sqrt{2} \cos \omega} \left (
\begin{matrix}
e^{\mathrm{i}\omega}\\
1
\end{matrix} \right ),\\
e^*_- (\sigma^\omega_x) & = -\frac{1}{\sqrt{2} \cos \omega} \left (
\begin{matrix}
- e^{-\mathrm{i}\omega}\\
1
\end{matrix} \right ),
\end{split}
\right.\eeq
where we have written $e_+$ for $e_1$ and $e_-$ for $e_2.$ Note that
\be
\langle e_- (\sigma^\omega_x), e_+ (\sigma^\omega_x) \rangle = \frac{1}{2} ( 1 - e^{2\mathrm{i}\omega}) \not=0
\ee
for $\omega \not=0,$ namely if $\omega \not=0,$ $e_+$ and $e_-$ are not orthogonal, due to the fact that $\sigma^\omega_x$ is not a self-adjoint operator.

Define
\be
G= |e^*_+ (\sigma^\omega_x) \rangle \langle e^*_+ (\sigma^\omega_x)| + |e^*_- (\sigma^\omega_x)\rangle \langle e^*_- (\sigma^\omega_x)| = \frac{1}{ \cos^2 \omega} \left (
\begin{matrix}
1 & \mathrm{i} \sin \omega\\
-\mathrm{i}\sin \omega & 1
\end{matrix} \right )
\ee
which has eigenvalues $\lambda_\pm = \frac{1}{\cos^2 \omega} (1 \pm \sin \omega),$ whose eigenstates are respectively
\be
e_+ (G) = \frac{1}{\sqrt{2}} \left (
\begin{matrix}
1\\
- \mathrm{i}
\end{matrix} \right ), \quad
e_- (G) = \frac{1}{\sqrt{2}} \left (
\begin{matrix}
- \mathrm{i}\\
1
\end{matrix} \right ).
\ee
Then
\be\begin{split}
G^\frac{1}{2} = & \frac{(1 + \sin \omega)^\frac{1}{2}}{\cos \omega} |e_+ (G) \rangle \langle e_+ (G)| + \frac{(1 - \sin \omega)^\frac{1}{2}}{\cos \omega} |e_- (G)\rangle \langle e_- (G)|\\
= & \frac{(1 + \sin \omega)^\frac{1}{2}}{2 \cos \omega} \left (
\begin{matrix}
1 & \mathrm{i} \\
-\mathrm{i} & 1
\end{matrix} \right )
+
\frac{(1 - \sin \omega)^\frac{1}{2}}{2\cos \omega} \left (
\begin{matrix}
1 & - \mathrm{i} \\
\mathrm{i} & 1
\end{matrix} \right )
\end{split}\ee
and
\be
G^{-\frac{1}{2}}= \frac{\cos \omega}{2(1 + \sin \omega)^\frac{1}{2}} \left (
\begin{matrix}
1 & \mathrm{i} \\
-\mathrm{i} & 1
\end{matrix} \right ) + \frac{\cos \omega}{2(1 - \sin \omega)^\frac{1}{2}} \left (
\begin{matrix}
1 & - \mathrm{i} \\
\mathrm{i} & 1
\end{matrix} \right ).
\ee
Thus,
\be
G^\frac{1}{2} \sigma^\omega_x G^{-\frac{1}{2}}= \frac{(1 - \sin \omega)^\frac{1}{2}}{\cos \omega} \left (
\begin{matrix}
0 & 1 \\
1 & 0
\end{matrix} \right ),
\ee
and so $\sigma^\omega_x$ is a para-Hermitian operator in $\mathbb{C}^2.$ Then, for a state determined by
\beq\label{eq:psi}
\psi = \cos \frac{\theta}{2} \left (
\begin{matrix}
1 \\
0
\end{matrix} \right ) + e^{\mathrm{i} \phi}\sin \frac{\theta}{2} \left (
\begin{matrix}
0 \\
1
\end{matrix} \right )
= \left (
\begin{matrix}
\cos \frac{\theta}{2} \\
 e^{\mathrm{i} \phi}\sin \frac{\theta}{2}
\end{matrix} \right ), \quad \theta,\phi \in [0, 2 \pi),
\eeq
the expectation of $\sigma^\omega_x$ at $\psi$ is
\be
\langle \sigma^\omega_x \rangle_{\psi, G} = \frac{(1 - \sin \omega)^\frac{1}{2}}{\cos \omega} \sin \theta \cos \phi
\ee
under the measurement $\{ |e_+ (\sigma^\omega_x) \rangle \langle e^*_+ (\sigma^\omega_x)|, |e_- (\sigma^\omega_x)\rangle \langle e^*_- (\sigma^\omega_x)|\}$ or equivalently in the measurement context of $G.$

Since $\sigma^\omega_y = \sigma_y,$ two eigenstates of it are
\beq\label{eq:Yeigenstate}
\left \{\begin{split}
e_+ (\sigma^\omega_y)& =\frac{1}{\sqrt{2}} \left (
\begin{matrix}
1 \\
\mathrm{i}
\end{matrix} \right )= e^*_+ (\sigma^\omega_y),\\
e_- (\sigma^\omega_y)&= \frac{1}{\sqrt{2}} \left (
\begin{matrix}
\mathrm{i} \\
1
\end{matrix} \right )= e^*_- (\sigma^\omega_y).
\end{split}
\right.\eeq
Then the expectation of $\sigma^\omega_x$ at $\psi$ is
\be
\langle \sigma^\omega_y \rangle_\psi = \langle \sigma^\omega_y \rangle_{\psi, G} = \sin \theta \sin \phi
\ee
in the measurement context of $G,$ where $G$ is any strictly positive self-adjoint operator such that $G^\frac{1}{2} \sigma_y G^{-\frac{1}{2}}$ is self-adjoint in $\mathbb{C}^2$ or equivalently $G$ commutes with $\sigma_y.$

Finally, the eigenstates of $\sigma^\omega_z$ are
\beq\label{eq:Zeigenstate}
\left \{\begin{split}
e_+ (\sigma^\omega_z) & = \frac{(1+ \cos \omega)^\frac{1}{2}}{\sqrt{2}}\left (
\begin{matrix}
1 \\
\frac{\mathrm{i}\sin \omega}{1+ \cos \omega}
\end{matrix} \right ),\\
e_- (\sigma^\omega_z) & = \frac{(1+ \cos \omega)^\frac{1}{2}}{\sqrt{2}} \left (
\begin{matrix}
-\frac{\mathrm{i}\sin \omega}{1+ \cos \omega} \\
1
\end{matrix} \right ),
\end{split}
\right.\eeq
and so,
\beq\label{eq:ZstarEigenstate}
\left \{\begin{split}
e^*_+ (\sigma^\omega_z) & =\frac{(1+ \cos \omega)^\frac{1}{2}}{\sqrt{2} \cos \omega} \left (
\begin{matrix}
1 \\
-\frac{\mathrm{i}\sin \omega}{1+ \cos \omega}
\end{matrix} \right ),\\
e^*_- (\sigma^\omega_z) & =\frac{(1+ \cos \omega)\frac{1}{2}}{\sqrt{2} \cos \omega} \left (
\begin{matrix}
\frac{\mathrm{i}\sin \omega}{1+ \cos \omega}\\
1
\end{matrix} \right ).
\end{split}
\right.\eeq
Note that
\be
\langle e_- (\sigma^\omega_z), e_+ (\sigma^\omega_z) \rangle =\mathrm{i} \sin \omega \not=0
\ee
if $\omega \not=0,$ in which case $\sigma^\omega_z$ is not a self-adjoint operator. Define
\be
G = |e^*_+ (\sigma^\omega_z) \rangle \langle e^*_+ (\sigma^\omega_z)| + |e^*_- (\sigma^\omega_z)\rangle \langle e^*_- (\sigma^\omega_z)| = \frac{1}{\cos^2 \omega} \left (
\begin{matrix}
1 & \mathrm{i} \sin \omega\\
-\mathrm{i}\sin \omega & 1
\end{matrix} \right )
\ee
which is the same $G$ as appearing in the case of $\sigma^\omega_x.$ Then
\be
G^\frac{1}{2} \sigma^\omega_z G^{-\frac{1}{2}}= \left (
\begin{matrix}
1 & 0 \\
0 & -1
\end{matrix} \right ),
\ee
and so $\sigma^\omega_z$ is a para-Hermitian operator in $\mathbb{C}^2.$ Therefore, the expectation of $\sigma^\omega_z$ at $\psi$ is
\be
\langle \sigma^\omega_z \rangle_{\psi, G} = \langle \sigma_z \rangle_\psi = \cos \theta
\ee
under the measurement $\{ |e_+ (\sigma^\omega_z) \rangle \langle e^*_+ (\sigma^\omega_z)|, |e_- (\sigma^\omega_z)\rangle \langle e^*_- (\sigma^\omega_z)|\}$ or equivalently in the measurement context of $G.$
\end{example}

\begin{example}\label{ex:qubitParaHermitian}\rm
Consider a two-dimensional Hamiltonian of the form
\beq\label{eq:2-dParaHamiltonian}
H = \left (
\begin{matrix}
r e^{\mathrm{i}\theta} & \gamma\\
\gamma & r e^{-\mathrm{i}\theta}
\end{matrix} \right )
\eeq
where the three parameters $r, \theta, $ and $\gamma$ are real numbers. Then $H$ has two real eigenvalues $\lambda_\pm = r \cos \theta \pm \sqrt{\gamma^2 - r^2 \sin^2 \theta}$ provided that $\gamma^2 > r^2 \sin^2 \theta,$ and the associated eigenstates of $H$ are
\be
e_+ (H) = \frac{1}{\sqrt{2}}\left (
\begin{matrix}
e^{\mathrm{i} \phi /2} \\
e^{-\mathrm{i} \phi /2}
\end{matrix} \right ),\quad
e_- (H) = \frac{1}{\sqrt{2}}\left (
\begin{matrix}
\mathrm{i} e^{-\mathrm{i} \phi /2} \\
-\mathrm{i} e^{\mathrm{i} \phi /2}
\end{matrix} \right ),
\ee
where the real number $\phi$ is defined by $\sin \phi = \frac{r}{\gamma} \sin \theta,$ and so
\be
e^*_+ (H)= \frac{1}{\sqrt{2}\cos \phi}\left (
\begin{matrix}
e^{-\mathrm{i} \phi /2} \\
e^{\mathrm{i} \phi /2}
\end{matrix} \right ),\quad
e^*_- (H)= \frac{1}{\sqrt{2}\cos \phi}\left (
\begin{matrix}
\mathrm{i} e^{\mathrm{i} \phi /2} \\
-\mathrm{i} e^{-\mathrm{i} \phi /2}
\end{matrix} \right ).
\ee
Note that
\be
\langle e_- (H), e_+ (H) \rangle = 2 \sin \phi \not=0,
\ee
when $r, \theta \not=0,$ in this case $H$ is not a self-adjoint operator.

Define
\be
G = |e^*_+ (H) \rangle \langle e^*_+ (H)| + |e^*_- (H)\rangle \langle e^*_- (H)| = \frac{1}{\cos^2 \phi} \left (
\begin{matrix}
1 & -\mathrm{i}\sin \phi\\
\mathrm{i}\sin \phi & 1
\end{matrix} \right )
\ee
which is the same $G$ ($\phi = -\omega$) as appearing in the case of $\sigma^\omega_x$ in Example \ref{ex:para-PaulMat}. Then
\be
G^\frac{1}{2} H G^{-\frac{1}{2}}= \cos \theta \left (
\begin{matrix}
r & 0 \\
0 & r
\end{matrix} \right ) + |\cos \phi| \gamma \left (
\begin{matrix}
0 & 1 \\
1 & 0
\end{matrix} \right ),
\ee
and so $H$ is a para-Hermitian operator in $\mathbb{C}^2$ provided $r, \theta \not=0$ such that $\gamma^2 > r^2 \sin^2 \theta.$ Note that, for $\gamma =1$ and $\theta = \frac{\pi}{2},$ the non-Hermitian operator $H =\sigma_r = \sigma_x - \mathrm{i} r \sigma_z$ is para-Hermitian for $0 < r < 1.$

Moreover, the expectation of $H$ at $\psi = \left (
\begin{matrix}
a \\
b
\end{matrix} \right )$ ($|a|^2 + |b|^2 =1$) is
\be
\langle H \rangle_{\psi, G} = r \cos \theta + \gamma |\cos \phi| (a \bar{b} + \bar{a} b)
\ee
under the measurement $\{ |e_+ (H) \rangle \langle e^*_+ (H)|, |e_- (H)\rangle \langle e^*_- (H)|\}$ or equivalently in the measurement context of $G.$

\end{example}

\section{PT-symmetric and biorthogonal quantum mechanics}\label{PTBIqm}

This section shows how the above formalism can recover PT-symmetric and biorthogonal quantum mechanics.

\subsection{PT-symmetric quantum mechanics}\label{PT-qm}

Recall that the linear parity operator $\mathcal{P}$ is defined by $\mathcal{P} f (x) = f(-x),$ whereas the anti-linear time-reversal operator $\mathcal{T}$ is defined by $\mathcal{T} f (x) = \overline{f(x)},$ so that $\mathcal{P}\mathcal{T} f (x) = \overline{f(-x)}.$ In \cite{BB1998}, Bender and Boettcher found that non-Hermitian Hamiltonian $H = p^2 + x^2 (\mathrm{i} x)^\nu$ ($\nu \ge 0$) may have real and positive spectrum, by noticing that while $H$ is not symmetric under $\mathcal{P}$ or $\mathcal{T}$ separately, it is invariant under their combined operation $\mathcal{P}\mathcal{T},$ that is, $H$ is $\mathcal{P}\mathcal{T}$-symmetric. The reality of the spectrum of $H$ is a consequence of its unbroken $\mathcal{P}\mathcal{T}$ symmetry. The unbroken $\mathcal{P}\mathcal{T}$ symmetry of $H$ means that eigenfunctions of $H$ are also eigenfunctions of $\mathcal{P}\mathcal{T}.$ In this case, there appears to be a natural choice for an inner produce given by
 \be
 (f, g)_{\mathcal{P}\mathcal{T}} = \int [\mathcal{P}\mathcal{T} f (x)] g (x) d x
 \ee
such that the eigenfunctions $\{\psi_n\}$ of $H$ are orthogonal, i.e., $(\psi_m, \psi_n)_{\mathcal{P}\mathcal{T}} = (-1)^n \delta_{mn}.$ However, the $\mathcal{P}\mathcal{T}$ inner product $ (f, g)_{\mathcal{P}\mathcal{T}}$ is indefinite, because the norms of the eigenfunctions via this inner product may be negative.

In \cite{BBJ2002}, Bender {\it et al.} observed that for any $\mathcal{P}\mathcal{T}$-symmetric Hamiltonian $H$ having an unbroken $\mathcal{P}\mathcal{T}$ symmetry, there exists a symmetry of $H$ described by a linear operator $\mathcal{C}$ defined by
\be
 C(x,y)= \sum_n \psi_n(x) \psi_n(y),
 \ee
such that the eigenfunctions $\{\psi_n\}$ of $H$ are orthogonal in the $\mathcal{C}\mathcal{P}\mathcal{T}$-inner product
 \be
 (f, g)_{\mathcal{C}\mathcal{P}\mathcal{T}} = \int [\mathcal{C}\mathcal{P}\mathcal{T} f (x)] g (x) d x,
 \ee
that is, $(\psi_m, \psi_n)_{\mathcal{C}\mathcal{P}\mathcal{T}} = \delta_{mn}$ and so the inner product $(f, g)_{\mathcal{C}\mathcal{P}\mathcal{T}}$ is positively definite. This means that any $\mathcal{P}\mathcal{T}$-symmetric Hamiltonian $H$ having an unbroken $\mathcal{P}\mathcal{T}$ symmetry has a metric operator $G =\mathcal{C}\mathcal{P}\mathcal{T}$ such that $H$ is a Hermitian operator in the $\mathcal{C}\mathcal{P}\mathcal{T}$ inner product $(f, g)_{\mathcal{C}\mathcal{P}\mathcal{T}}.$ Then, by Proposition \ref{prop:paraHop}, we have

\begin{proposition}\label{prop:PTqm}\rm
If a densely defined closed $\mathcal{P}\mathcal{T}$-symmetric operator $H$ with discrete spectrum has an unbroken $\mathcal{P}\mathcal{T}$ symmetry, then $H$ is a para-Hermitian operator.
\end{proposition}

Mathematically, the $\mathcal{P}\mathcal{T}$ symmetry of a Hamiltonian provides a way to construct a metric operator for the Hamiltonian. Once having the metric operator $G$ associated with a $\mathcal{P}\mathcal{T}$-symmetric Hamiltonian $H,$ the measurement and dynamics relative to $H$ can be resolved in the mathematical formwork of Definition \ref{df:MathFNHQM}. In this sense, we say that the mathematical formalism of non-Hermitian quantum mechanics given by Definition \ref{df:MathFNHQM} recovers $PT$-symmetric quantum mechanics.

For illustration, in what follows, we show the fact, which was previously found by Bender {\it et al.} \cite{BBJM2007} in the formwork of $PT$-symmetric quantum mechanics, that the transformation between a pair of orthogonal states according to non-Hermitian quantum mechanics in the formwork of Definition \ref{df:MathFNHQM} can be arbitrarily faster than Hermitian quantum mechanics under the same energy constraint.

To this end, for the sake of clarity, let us consider the transformation from the initial state $|0\rangle = \left (
\begin{matrix}
1 \\
0
\end{matrix} \right )$ to the final state $|1\rangle = \left (
\begin{matrix}
0 \\
1
\end{matrix} \right )$ in a qubit system associated with the two-dimensional Hilbert space $\mathbb{C}^2.$ As shown in \cite{BBJM2007}, for any two-dimensional Hermitian Hamiltonian $H,$ the smallest time required to transform $|0\rangle$ to $|1\rangle$ is $\tau = \frac{\pi \hbar}{\omega},$ i.e., $|1\rangle = \alpha e^{- \mathrm{i}H\tau /\hbar} |0\rangle$ with some phase factor $\alpha \not= 0,$ where $\omega$ is the difference between the energy eigenvalues of the Hamiltonian.

However, for a para-Hermitian Hamiltonian of the form
\be\label{eq:2-dParaHamiltonian}
H = \left (
\begin{matrix}
r e^{\mathrm{i}\theta} & \gamma\\
\gamma & r e^{-\mathrm{i}\theta}
\end{matrix} \right )
\ee
as shown in Example \ref{ex:qubitParaHermitian}, one has
\be
e^{- \mathrm{i}H t /\hbar} |0\rangle = \frac{e^{- \mathrm{i} r t \cos \theta /\hbar}}{\cos \phi} \left (
\begin{matrix}
\cos (\frac{\omega}{2 \hbar} t - \phi) \\
-\mathrm{i} \sin(\frac{\omega t}{2 \hbar})
\end{matrix} \right ),
\ee
where $\omega = 2 (\gamma^2 - r^2 \sin^2 \theta )^\frac{1}{2}$ is the difference between the eigenvalues of $H,$ and $\phi$ is defined by $\sin \phi = \frac{r}{\gamma} \sin \theta.$ Then we see that the evolution time to reach $|1\rangle$ from $|0\rangle$ is $t = (2 \phi + \pi)\hbar /\omega.$ Taking allowable values for $r, \gamma,$ and $\theta,$ we can make $\phi$ approaching $- \pi/2$ such that the optimal time $\tau$ tends to $0.$ Thus, the transformation from the initial state $|0\rangle$ to the final state $|0\rangle$ according to non-Hermitian quantum mechanics in the formwork of Definition \ref{df:MathFNHQM} can be arbitrarily faster than Hermitian quantum mechanics under the same energy constraint.

\subsection{Biorthogonal quantum mechanics}\label{Bi-qm}

Given an unconditional basis $\{e_n\}$ in a Hilbert space $\mathbb{H}$ with the unique dual basis $\{ e^*_n\}$ ($\{e_n, e^*_n\}$ is called a biorthogonal basis, cf. \cite{LT1977}), a densely defined closed operator $T$ in $\mathbb{H}$ can be expressed by
\beq\label{eq:BiOpRep}
T = \sum_{n,m} f_{n m} |e_n \rangle \langle e^*_m |
\eeq
with $f_{n m} = \langle e^*_n, T e_m \rangle,$ provided $\{e_n\} \subset \mathcal{D} (T).$ In biorthogonal quantum mechanics \cite{Brody2014}, such a operator $T$ is said to be a `biorthogonally Hermitian' operator with respect to the biorthogonal basis $\mathcal{F} =\{e_n, e^*_n\},$ if $\bar{f}_{n m} = f_{m n}$ for any $n,m.$

Given a biorthogonal basis $\mathcal{F}=\{e_n, e^*_n\},$ as assumed in \cite{Brody2014}, each observable is represented by a `biorthogonally Hermitian' operator $T$ relative to $\mathcal{F},$ and the expectation of $T$ at a state $\psi$ is defined by
\beq\label{eq:BiOpExp}
\langle T \rangle_{\psi, \mathcal{F}}  = \frac{\langle \tilde{\psi}, T \psi \rangle}{\langle \tilde{\psi}, \psi \rangle},
\eeq
where $\tilde{\psi} = \sum_n a_n e^*_n$ with $a_n = \langle e^*_n, \psi \rangle$ (noticing that $\psi = \sum_n a_n e_n$). Then $\langle T \rangle_{\psi, \mathcal{F}}$ defined by \eqref{eq:BiOpExp} is real for any `biorthogonally Hermitian' operator $T$ relative to a biorthogonal basis $\mathcal{F}$ and for all states $\psi,$ since
\be
\langle T \rangle_{\psi, \mathcal{F}}  = \frac{\sum_{n,m} \bar{a}_n a_m f_{n m}}{\sum_n |a_n|^2}.
\ee
Note that a `biorthogonally Hermitian' operator $T$ relative to a biorthogonal basis $\mathcal{F} =\{e_n, e^*_n\}$ is not necessarily Hermitian if $\{e_n\}$ is not an orthogonal basis, and so $\langle \psi, T \psi \rangle/ \langle \psi, T \psi \rangle$ is not real for most states $\psi,$ as noted in \cite{Brody2014}.

\begin{proposition}\label{prop:BiHop}\rm
Let $\mathcal{F} =\{e_n, e^*_n\}$ be a biorthogonal basis in $\mathbb{H}.$ A densely defined closed operator $T$ with $\{e_n\} \subset \mathcal{D} (T)$ is a `biorthogonally Hermitian' operator $T$ relative to $\mathcal{F}$ if and only if $T$ is self-adjoint with respect to the inner product $\langle \cdot, \cdot \rangle_G,$ where $G = \sum_n |e^* \rangle \langle e^*_n|$ is a metric operator associated with $T.$ Consequently, if a densely defined closed operator is `biorthogonally Hermitian' then it is a para-Hermitian operator. Moreover,
\be
\langle T \rangle_{\psi, \mathcal{F}} = \langle T \rangle_{G^\frac{1}{2} \psi, G}
\ee
for any nonzero $\psi \in \mathbb{H}.$
\end{proposition}

\begin{proof}
Note that
\be
\langle u, T v \rangle_G = \langle u, \sum_{n,m} f_{n m} |e^*_n \rangle \langle e^*_m | v \rangle = \langle \sum_{n,m} \bar{f}_{n m} |e_m \rangle \langle e^*_n |u, G v \rangle = \langle T u, v \rangle_G
\ee
for any $u,v \in \mathbb{H},$ whenever $\bar{f}_{n m} = f_{m n}$ for any $n,m.$ This concludes the first assertion.

For the second assertion, since
\be
\langle \psi, T \psi \rangle_G = \langle \psi, G T \psi \rangle = \sum_{n,m} \bar{a}_n a_m f_{n m}
\ee
for $\psi = \sum_n a_n e_n,$ and $\|G^\frac{1}{2} \psi\| = \sum_n |a_n|^2$ (see Remark \ref{rk:NHQM}), this follows the required equality.
\end{proof}

By this proposition, we see that a densely defined closed `biorthogonally Hermitian' operator is a para-Hermitian operator, and the expectation value formula \eqref{eq:BiOpExp} reduces to the non-Hermitian Born formula \eqref{eq:BornRuleNH}. In this sense, we say that the mathematical formalism of non-Hermitian quantum mechanics given by Definition \ref{df:MathFNHQM} recovers biorthogonal quantum mechanics.

\section{Observable-geometric phase}\label{GeoPhase}

The notion of the geometric phase for a quantum system was introduced by Berry (cf. \cite{Berry1984, Simon1983}), on which there exist extensive works (cf. \cite{AA1987, BMKNZ2003, PLC2022, Sj2015} and references therein). The geometric phases for non-Hermitian quantum systems have been studied in \cite{CZ2012, GW1988, MM2008, SB1988, ZW2019}, etc. As usual, these geometric phases are associated with the quantum state. Recently, the notion of the geometric phase for the observable (the so-called observable-geometric phase) was introduced in \cite{Chen2020}, which is defined as a sequence of phases associated with a complete set of eigenstates of the observable. In this section, we will study the observable-geometric phase in the non-Hermitian setting via the mathematical framework of non-Hermitian quantum mechanics given by Definition \ref{df:MathFNHQM}. We first define the notion of an observable-geometric phase in the non-Hermitian case using the evolution system mentioned in Section \ref{Pre:EvoSys}. Then we give the geometric interpretation of it using the geometry of the non-Hermitian observable space based on the group $\mathcal{T} (\mathbb{H})$ of invertible bounded operators (see Section \ref{App} for the details).

Consider a non-Hermitian quantum system with a time-dependent Hamiltonian $\{h(t): t \in [0,T]\},$ where $h(t)$'s are all para-Hermitian operators. Assume that $h(t)$'s have the same domain $\mathbb{D}$ and there exists an evolution system $\{U(t,s) \in \mathcal{T} (\mathbb{H}): t, s \in [0,T]\}$ for $A(t) = - \mathrm{i}h(t)$ on $\mathbb{D}.$ By \eqref{eq:EvoSysEqu} and $U(t,s)^{-1} = U(s,t),$ we then have
the Schr\"{o}dinger equation
\begin{equation}\label{eq:SchrodingerEquPropagatorI}
\mathrm{i} \frac{d}{d t} U(t,s)\phi = h(t) U(t,s)\phi, \quad \forall \phi \in \mathbb{D},
\end{equation}
and the skew Schr\"{o}dinger equation
\begin{equation}\label{eq:SchrodingerEquPropagatorII}
\mathrm{i} \frac{d}{d t} U(s,t) \phi = - \tilde{h}_s(t)U(s,t) \phi, \quad \forall \phi \in \mathbb{D},
\end{equation}
where $\tilde{h}_s(t) = U(s,t) h(t) U(t, s).$ In this case, the evolution system $\{U(t,s): t, s \in [0,T]\}$ is also called the {\it time evolution operator} or {\it propagator} generated by $h(t)$ (cf. \cite[X.69]{RS1980II}).

Note that, by \eqref{eq:SchrodingerEquPropagatorI}, for any $s \in [0,T)$ and $\phi \in \mathbb{H},$ $\phi_s (t) = U (t,s) \phi$ is the unique solution of the time-dependent Schr\"{o}dinger equation
\begin{equation}\label{eq:SchrodingerEquTime}
\mathrm{i} \frac{d}{d t} \phi_s(t) = h(t) \phi_s (t),\quad \phi_s (s) = \phi.
\end{equation}
Given any observable $X_0,$ namely a para-Hermitian operator on $\mathbb{H},$ by \eqref{eq:SchrodingerEquPropagatorI} and \eqref{eq:SchrodingerEquPropagatorII} we conclude that $X(t) = U(0,t) X_0 U(t,0)$ is the unique solution of the time-dependent Heisenberg equation
\begin{equation}\label{eq:HeisenbergEquTime}
\mathrm{i} \frac{d X (t)}{d t} = [X (t), \tilde{h}(t)],\quad X(0) = X_0,
\end{equation}
where $\tilde{h}(t) = U(0,t) h(t) U(t, 0).$ If there exists $\tau \in (0, T)$ such that $X(\tau) = X(0),$ the time evolution of observable $X(t)$ is then called {\it cyclic} with period $\tau,$ and $X_0 = X(0)$ is said to be a cyclic observable.

Suppose that the observable $X_0$ has a non-degenerate eigenvalue associated with every eigenstate $\psi_n \in \mathbb{D}$ for $n \ge 1,$ and $X(t) = U(0,t) X_0 U(t,0)$ is cyclic with period $\tau \in (0,T),$ namely $X(\tau) = X_0.$ Then $U(0, \tau)\psi_n = e^{\mathrm{i} \theta_n}\psi_n$ with some $\theta_n \in \mathbb{C}$ for $n \ge 1.$ Denoting $\psi_n (t) = U(0,t) \psi_n$ for $n \ge 1,$ which are the eigenstates of $X(t),$ by \eqref{eq:SchrodingerEquPropagatorII} we conclude that $\psi_n (t)$ satisfies the skew (time-dependent) Schr\"{o}dinger equation
\begin{equation}\label{eq:SchrodingerEquTimeEigenstate}
\mathrm{i} \frac{d}{d t} \psi_n (t) = - \tilde{h}(t) \psi_n (t),\quad \psi_n (0) = \psi_n.
\end{equation}
Note that $\{\psi_n (t): n \ge 1\}$ is an unconditional basis in $\mathbb{H}$ for every $t.$ Let $\psi^*_n (t) = U^*(t,0) \psi^*_n$ for $n \ge 1,$ then $\{\psi^*_n (t): n \ge 1\}$ is the dual basis of $\{\psi_n (t): n \ge 1\}$ such that $\langle \psi^*_n (t),\psi_m (t) \rangle = \delta_{n m}$ (see \cite{LT1977} for the details of the unconditional basis).

For each $n \ge 1,$ we define
\begin{equation}\label{eq:Parallelvect}
|\tilde{\psi}_n (t) \rangle = e^{ - \mathrm{i} \int^t_0 \langle \psi^*_n (s) | \tilde{h} (s) | \psi_n (s) \rangle d s } |\psi_n (t) \rangle.
\end{equation}
Since
\be
|\tilde{\psi}_n (t) \rangle = e^{ - \mathrm{i} \int^t_0 \langle \psi^*_n (0) | h (s) | \psi_n (0) \rangle d s } |\psi_n (t) \rangle,
\ee
then $|\tilde{\psi}_n (\tau) \rangle = e^{\mathrm{i} \beta_n}|\psi_n (0) \rangle,$ where
\beq\label{eq:ObGP-NH}
\beta_n = \theta_n - \int^\tau_0 \langle \psi^*_n (0) | h(t) | \psi_n (0)\rangle d t.
\eeq
Moreover, from \eqref{eq:SchrodingerEquTimeEigenstate} we conclude
\begin{equation}\label{eq:ParallelCondVect}
\langle \tilde{\psi}_n (t) | \frac{d}{d t} |\tilde{\psi}_n (t) \rangle =0.
\end{equation}

Also, for any closed path
\begin{equation}\label{eq:ClosedCurve}
|\bar{\psi}_n (t) \rangle =  e^{- \mathrm{i} \alpha_n (t)} |\psi_n (t) \rangle,
\end{equation}
where $\alpha_n: [0, \tau) \mapsto \mathbb{C}$ is continuously differential and $\alpha_n (\tau) - \alpha_n (0) = \theta_n$ for every $n \ge 1,$ i.e., $|\bar{\psi}_n (\tau) \rangle = |\bar{\psi}_n (0) \rangle,$ we have
\begin{equation}\label{eq:ObGP-ClosedCurve}
\beta_n = \int^\tau_0\mathrm{i} \langle \bar{\psi}^*_n (t) | \frac{d}{d t} |\bar{\psi}_n (t) \rangle d t,
\end{equation}
where $\bar{\psi}^*_n (t) = e^{-\mathrm{i} \alpha_n (t)} |\psi^*_n (t) \rangle$ for every $n \ge 1.$

Following \cite{Chen2020}, this leads to the notion of the observable-geometric phase in the non-Hermitian setting as follows.

\begin{definition}\label{df:ObGeoPhase}
Using the above notations, the observable-geometric phases of the periodic evolution of observable $X(t)$ in a non-Hermitian quantum system are defined by
\begin{equation}\label{eq:q-GeoPhase}
\beta_n = \theta_n - \int^\tau_0 \langle \psi^*_n (0) | h(t) | \psi_n (0)\rangle d t
\end{equation}
which is uniquely defined up to $2 \pi k$ ($k$ is integer) for every $n \ge 1.$
\end{definition}

\begin{remark}\rm
\begin{enumerate}[{\rm 1)}]

\item Note that for every $n \ge 1,$ $\beta_n$ may be a complex number (see Example \ref{ex:qubit} below). This is different from the ones of a Hermitian quantum system as defined in \cite{Chen2020}.

\item If $h(t)$'s are all Hermitian, then $\psi^*_n (t) = \psi_n (t)$ and the observable-geometric phases $\beta_n$'s are all real and coincide with the ones defined in \cite{Chen2020}.

\item When some eigenvalues of the initial observable $X_0$ are degenerate as eigenstates, this would lead to the notion of non-Abelian observable-geometric phase as similar to the usual non-Abelian geometric phase (cf. \cite{Anandan1988,STAHJS2012}). We will discuss it elsewhere.

\item We can also discuss the adiabatic case of the observable-geometric phase in the non-Hermitian setting, as done in \cite{Chen2020} in the Hermitian case. We omit the details.

\end{enumerate}
\end{remark}

For illustrating the observable-geometric phase in a non-Hermitian quantum system, we consider a qubit case, namely the Hilbert space $\mathbb{H} = \mathbb{C}^2.$

\begin{example}\label{ex:qubit}\rm
Consider a non-Hermitian qubit system, whose Hamiltonian is $H = -\sigma^\omega_z$ (see Example \ref{ex:para-PaulMat}). Given a spin observable $X_0$ with two non-degenerate eigenstates
\be\label{eq:InitalObsQubit}
\psi_1 = \left ( \begin{matrix} \cos \frac{\phi}{2} \\
\sin \frac{\phi}{2}
\end{matrix}\right ),\;  \psi_2 = \left ( \begin{matrix} -\sin \frac{\phi}{2} \\
\cos \frac{\phi}{2}
\end{matrix}\right )
\ee
in $\mathbb{C}^2,$ $X(t) = U(0,t) X_0 U(t,0)$ satisfies Eq.\eqref{eq:HeisenbergEquTime} with
\be
\tilde{h} (t) = h (t)= -\sigma^\omega_z
\ee
and $U(t,0) = e^{\mathrm{i}t \sigma^\omega_z}.$ Note that $\sigma^\omega_z$ has eigenvalues $1$ and $-1,$ and the corresponding eigenstates are
\be\label{eq:Zeigenstate}
\left \{ \begin{split}
e_+ (\sigma^\omega_z) & = \left (
\begin{matrix}
1 \\
\frac{\mathrm{i}\sin \omega}{1+ \cos \omega}
\end{matrix} \right ),\\
e_- (\sigma^\omega_z) & = \left (
\begin{matrix}
-\frac{\mathrm{i}\sin \omega}{1+ \cos \omega} \\
1
\end{matrix} \right ),
\end{split}\right.\ee
and so,
\be\label{eq:ZstarEigenstate}
\left \{ \begin{split}
e^*_+ (\sigma^\omega_z) & =\frac{1+ \cos \omega}{2 \cos \omega} \left (
\begin{matrix}
1 \\
-\frac{\mathrm{i}\sin \omega}{1+ \cos \omega}
\end{matrix} \right ),\\
e^*_- (\sigma^\omega_z) & =\frac{1+ \cos \omega}{2\cos \omega} \left (
\begin{matrix}
\frac{\mathrm{i}\sin \omega}{1+ \cos \omega}\\
1
\end{matrix} \right ).
\end{split}
\right.\ee
Then, by Definition \ref{df:FunctCalculusparaHop} we have
\be\begin{split}
U(t,0)& =e^{\mathrm{i}t} |e_+ (\sigma^\omega_z) \rangle \langle e^*_+ (\sigma^\omega_z)| + e^{-\mathrm{i} t} |e_- (\sigma^\omega_z) \rangle \langle e^*_- (\sigma^\omega_z)|,\\
U(0,t)& =e^{-\mathrm{i}t} |e_+ (\sigma^\omega_z) \rangle \langle e^*_+ (\sigma^\omega_z)| + e^{\mathrm{i} t} |e_- (\sigma^\omega_z) \rangle \langle e^*_- (\sigma^\omega_z)|.
\end{split}\ee
Define $\psi_n (t) = U(0,t) \psi_n$ for $n=1,2,$ we have
\be
\begin{split}
\psi_1 (t) = & \frac{e^{-\mathrm{i}t}}{2 \cos \omega} \left ( \begin{matrix} (1 + \cos \omega) \cos \frac{\phi}{2} + \mathrm{i} \sin \omega \sin \frac{\phi}{2} \\
- \frac{\sin^2 \omega \sin \frac{\phi}{2}}{1+\cos \omega} + \mathrm{i} \sin \omega \cos \frac{\phi}{2}
\end{matrix}\right ) + \frac{e^{\mathrm{i}t}}{2 \cos \omega} \left ( \begin{matrix} - \frac{\sin^2 \omega \cos \frac{\phi}{2}}{1+\cos \omega} - \mathrm{i} \sin \omega \sin \frac{\phi}{2} \\
(1 + \cos \omega) \sin \frac{\phi}{2} - \mathrm{i} \sin \omega \cos \frac{\phi}{2}
\end{matrix}\right ),\\
\psi_2 (t) = & \frac{e^{-\mathrm{i}t}}{2 \cos \omega} \left ( \begin{matrix} - (1 + \cos \omega) \sin \frac{\phi}{2} + \mathrm{i} \sin \omega \cos \frac{\phi}{2} \\
- \frac{\sin^2 \omega \cos \frac{\phi}{2}}{1+\cos \omega} - \mathrm{i} \sin \omega \sin \frac{\phi}{2}
\end{matrix}\right ) + \frac{e^{\mathrm{i}t}}{2 \cos \omega} \left ( \begin{matrix} \frac{\sin^2 \omega \sin \frac{\phi}{2}}{1+\cos \omega} - \mathrm{i} \sin \omega \cos \frac{\phi}{2} \\
(1 + \cos \omega) \cos \frac{\phi}{2} + \mathrm{i} \sin \omega \sin \frac{\phi}{2}
\end{matrix}\right ),
\end{split}\ee
which satisfies the skew Schr\"{o}dinger equation \eqref{eq:SchrodingerEquTimeEigenstate}, namely
$$
\mathrm{i} \frac{d \psi_n (t) }{d t} =\sigma^\omega_z \psi_n (t),\quad n=1,2.
$$
The evolution $X(t)$ is periodic with period $\tau = \pi,$ precisely $\psi_n (\pi) = e^{\mathrm{i} \pi} \psi_n (0)$ for $n=1,2.$

Since $\psi^*_n = \psi_n$ for $n=1,2,$ by \eqref{eq:q-GeoPhase} we have
\be
\beta_1 = \pi + \int^\pi_0 \langle \psi_1 | \sigma^\omega_z | \psi_1\rangle d t = \pi + \frac{\cos \phi + \mathrm{i} \sin \omega \sin \phi}{\cos \omega} \pi = \pi \big (1 + \frac{\cos \phi}{\cos \omega} \big ) +\frac{\mathrm{i} \pi \sin \omega \sin \phi}{\cos \omega},
\ee
and
\be
\beta_2 = \pi + \int^\pi_0 \langle \psi_2 | \sigma^\omega_z | \psi_2 \rangle d t = \pi \big (1 - \frac{\cos \phi}{\cos \omega} \big ) - \frac{\mathrm{i} \pi\sin \omega \sin \phi}{\cos \omega}.
\ee
Both are complex numbers if $\omega, \phi \not= 0.$
\end{example}

Finally, we give a geometric interpretation of $\beta_n$'s defined as in \eqref{eq:q-GeoPhase}, involving the geometry of the non-Hermitian observable space developed in Section \ref{App}.

Given a point $O_0 = \{|e_n \rangle \langle e_n^*|\}_{n \ge 1}$ in $\tilde{\mathcal{W}} (\mathbb{H}),$ using the above notations, we define $\tilde{V} (t) \in \mathcal{T} (\mathbb{H})$ for $0 \le t \le \tau$ by
$$
\tilde{V} (t) = \sum_{n \ge 1} |\tilde{\psi}_n (t) \rangle \langle e_n^* |,
$$
where $|\tilde{\psi}_n (t) \rangle$'s are defined in \eqref{eq:Parallelvect}. Then,
$$
\tilde{C}_P: \; [0, \tau] \ni t \longmapsto \tilde{V} (t) \in \mathcal{T} (\mathbb{H})
$$
is a smooth $O_0$-lift of $C_W: [0, \tau] \ni t \mapsto O(t) = \{ |\psi_n (t) \rangle \langle \psi_n^* (t)|\}_{n \ge 1}.$ Since $\tilde{V}^{-1} (t) = \sum_{n \ge 1} |e_n \rangle \langle \tilde{\psi}^*_n (t)|,$ by \eqref{eq:ParallelCondVect}, we have
\begin{equation}\label{eq:CanonicalParallelCond}
\check{\Omega}_{\tilde{V} (t)} \Big [ \frac{d \tilde{V} (t)}{d t} \Big ] = 0
\end{equation}
for all $t \in [0,\tau],$ where $\check{\Omega}$ is the canonical quantum connection (cf. Example \ref{Ex:CanonicalConnection}). This means that $[0, \tau] \ni t \mapsto \tilde{V} (t)$ is the parallel transportation along $C_W$ with respect to the {\it canonical connection} $\check{\Omega}$ on $\xi_{O_0}.$ Therefore,  $\tilde{C}_P$ is the {\it horizontal} $O_0$-lift of $C_W$ with respect to $\check{\Omega}$ in the principal bundle $\xi_{O_0}$ such that
\be
\tilde{V} (\tau) |e_n \rangle = |\tilde{\psi}_n (\tau) \rangle = e^{\mathrm{i} \beta_n} |\psi_n \rangle,\quad \forall n \ge 1,
\ee
and so
\begin{equation}\label{eq:HolonomyUnitaryOper}
\tilde{V} (\tau) = \sum^d_{n = 1} e^{\mathrm{i} \beta_n} | \psi_n \rangle \langle e_n^* |
\end{equation}
is the holonomy element associated with the connection $\check{\Omega}, C_W,$ and $V_0 = \sum_{n \ge 1}| \psi_n \rangle \langle e_n^* |$ in $\xi_{O_0}.$

In conclusion, we have the following theorem.

\begin{theorem}\label{thm:q-GeoPhase}\rm
\begin{enumerate}[\rm (1)]

\item For every $n \ge 1,$ the geometric phase $\beta_n$ defined in \eqref{eq:q-GeoPhase} is given by
\begin{equation}\label{eq:q-GeoPhaseExpression}
\beta_n =\langle e^*_n | \mathrm{i} \int^\tau_0 \check{\Omega}_{\bar{V} (t)} \Big [ \frac{d \bar{V} (t)}{d t} \Big ] d t |e_n \rangle = \langle e^*_n | \mathrm{i} \oint_{C_W} \bar{V}^{-1} \star d \bar{V} |e_n \rangle,
\end{equation}
where $\bar{C}_P: [0, \tau] \ni t \mapsto \bar{V} (t) \in \mathcal{T} (\mathbb{H})$ corresponds to any of the closed smooth $O_0$-lifts of $C_W$ with $\bar{V} (0) = \bar{V}_0,$ and $\check{\Omega}_V = V^{-1} \star d V$ is the canonical connection on $\xi_{O_0} (\mathbb{H}).$ Thus, $\beta_n$'s are independent of the choice of the time parameterization of $V (t),$ namely the speed with which $V (t)$ traverses its closed path. It is also independent of the choice of the Hamiltonian as long as the Heisenberg equations \eqref{eq:HeisenbergEquTime} involving these Hamiltonians describe the same closed path $C_W$ in $\tilde{\mathcal{W}} (\mathbb{H}).$

\item The set $\{ \beta_n: n \ge 1\}$ is independent of the choice of the starting point $V_0.$

\item The set $\{ \beta_n: n \ge 1\}$ is independent of the choice of the measurement point $O_0.$ Therefore, this number set is considered to be a set of geometric invariants for $C_W.$

\end{enumerate}
\end{theorem}

\begin{remark}\label{rk:q-ObsGeoPhase}\rm
In the definition \eqref{eq:q-GeoPhase}, the $\beta_n$'s are in fact independent of the choice of measurement points $O_0$ and background geometry over the fiber bundle $\xi_{O_0}.$
\end{remark}

\begin{proof}
(1).\; Let $\bar{C}_P: [0, \tau] \ni t \longmapsto \bar{V} (t) \in \mathcal{F}^{O(t)}_{O_0}$ be a smooth $O_0$-lift of $C_W$ such that $\bar{V} (\tau) = \bar{V} (0)=\bar{V}_0.$ By definition, $\bar{V} (t) = \sum_{n \ge 1} |\bar{\psi}_n (t) \rangle \langle e_n^* |$ and $\bar{V}^{-1} (t) = \sum_{n \ge 1} |e_n \rangle \langle \bar{\psi}^*_n (t) |,$ where $\bar{\psi}_n (\tau) = \bar{\psi}_n (0)$ for all $n \ge 1.$ By \eqref{eq:ObGP-ClosedCurve}, we
conclude \eqref{eq:q-GeoPhaseExpression}.

(2).\; For any $\check{V}_0 \in \mathcal{F}^{O(0)}_{O_0}$ there exists some $G = \sum_{n\ge 1} c_n| e_{\sigma (n)} \rangle \langle e^*_n | \in \mathcal{G}_{O_0}$ with $\sigma \in \Pi (d)$ and $(c_n)_{n \ge 1} \in \mathcal{G}_\8 (\mathbb{C}_*)$ such that $\check{V}_0 = V_0 G.$ Then $\check{C}_P: [0, T] \ni t \longmapsto \check{V}(t) = \tilde{V} (t) G$ is the horizontal $O_0$-lift of $C_W$ with the starting point $\check{V}(0) = V_0 G$ such that $\check{V}(T) |e_n\rangle = e^{\mathrm{i} \beta_{\sigma (n)}} \check{V} (0) |e_n\rangle$ for all $n \ge 1.$ Thus, the set $\{\beta_n: n \ge 1 \}$ is invariant for any starting point $V_0 \in \mathcal{F}^{O(0)}_{O_0}.$ Combining this fact with \eqref{eq:q-GeoPhaseExpression} yields
$$
\{\beta_n: n \ge 1 \} = \Big \{ \langle e^*_n | \mathrm{i} \oint_{C_W} \bar{V} \star d \bar{V} | e_n \rangle:\; n \ge 1 \Big \}
$$
for any closed smooth $O_0$-lift $\bar{C}_P$ of $C_W.$ Therefore, the observable-geometric phases are independent of the choice of the starting point and only depends on the geometry of the curve $C_W$ with respect to the $O_0$-connection $\check{\Omega}.$

(3).\; Let $\tilde{C}_P: [0, \tau] \ni t \longmapsto \tilde{V} (t)$ be the horizontal $O_0$-lift of $C_W$ with respect to $\Omega$ with the starting point $\tilde{V} (0) = V_0.$ For any $O'_0 = \{ | e_n' \rangle \langle (e'_n)^*|: n \ge 1\} \in \tilde{\mathcal{W}} (\mathbb{H})$ there exists some $T \in \mathcal{T} (\mathbb{H})$ such that $O'_0= T O_0 T^{-1}$ with $| e_n' \rangle = T | e_n \rangle$ for $n\ge 1.$ Then $\Omega' = \{\Omega'_P: P \in \mathcal{T} (\mathbb{H}) \}$ is a $O'_0$-connection on $\xi_{O'_0},$ where $\Omega'_P (Q) = T \Omega_{P T} (Q T) T^{-1}$ for any $P \in \mathcal{T} (\mathbb{H})$ and for all $Q \in T_P \xi_{O'_0} (\mathbb{H}).$ By computation, we conclude that $\tilde{C}'_P: [0, T] \ni t \longmapsto \tilde{V}' (t) = \tilde{V} (t) T^{-1}$ is the horizontal $O'_0$-lift of $C_W$ with respect to $\Omega'$ with the starting point $\tilde{V}' (0) = V_0 T^{-1}.$ Therefore,
$$
\tilde{V}' (\tau) | e_n' \rangle = \tilde{V} (\tau) T^{-1} | e_n' \rangle = \tilde{V} (\tau) | e_n \rangle = e^{\mathrm{i} \beta_n} \tilde{V} (0) | e_n \rangle = e^{\mathrm{i} \beta_n} \tilde{V}' (0) | e_n' \rangle,
$$
and hence the set of the geometric phases of $C_W$ with respect to $\Omega'$ is the same as that of $\Omega.$
\end{proof}
\begin{remark}\label{rk:RH}\rm
In a connection between quantum mechanics and mathematics, there is the Hilbert-P\'{o}lya conjecture that the nontrivial zeros of the Riemann zeta function correspond to the eigenvalues of a quantum mechanical Hamiltonian (cf. \cite{Sierra2019} and references therein). That the Hamiltonian is once proved to be Hermitian will confirm the Riemann hypothesis (cf. \cite{BBM2017}). Concerning with the Hilbert-P\'{o}lya conjecture, we would conjecture that the nontrivial zeros of the Riemann zeta function correspond to the observable-geometric phases of a cyclic (non-Hermitian) quantum evolution. This should shed light on the understanding of the Riemann hypothesis from the physical point of view.
\end{remark}
\section{Summary}\label{Sum}

Based on a theorem of Antoine and Trapani \cite{AT2014}, we introduce the notions of para-Hermitian and para-unitary operators, and prove a Stone type theorem for the one-parameter group of the para-unitary operators. In terms of the para-Hermitian and para-unitary operators, we present a mathematical formalism of non-Hermitian quantum mechanics, including the five postulates: the state postulate, the observable postulate, the measurement postulate, the evolution postulate, and the composite-systems postulate. These postulates are non-Hermitian analogies of those found in the Dirac-von Neumann formalism of quantum mechanics (the state and composite-systems postulates are the same in both formalisms). In particular, in the measurement postulate, we give the Born formula in the non-Hermitian setting. Indeed, our formalism is an extension of the Dirac-von Neumann formalism of quantum mechanics to the non-Hermitian setting. In the framework of this formalism, we generalize the notion of the observable-geometric phase \cite{Chen2020} to the non-Hermitian setting. We hope this formalism could play a role of a mathematical foundation for non-Hermitian quantum mechanics and its application to quantum computation and quantum information theory.

\section{Appendix: Geometry of non-Hermitian observable space}\label{App}

\subsection{Non-Hermitian observable space}\label{ObSpaceNH}

A complete decomposition in $\mathbb{H}$ is defined as a set $\{ | n \rangle \langle n^*|: n \ge 1 \}$ of projections of rank one satisfying
\begin{equation}\label{eq:OrthDecomp}
\sum_{n \ge 1} | n \rangle \langle n^*| = I,\quad \langle n^* |m \rangle = \delta_{n m}.
\end{equation}
We denote by $\tilde{\mathcal{W}} (\mathbb{H})$ the set of all complete decompositions in $\mathbb{H}.$ Note that a complete decomposition $O= \{ | n \rangle \langle n^*|: n \ge 1 \}$ determines uniquely a unconditional basis $\{| n \rangle \}_{n \ge 1}$ up to phases for basic vectors (we refer to \cite{LT1977} for the details of the unconditional basis). Conversely, a unconditional basis uniquely defines a complete decomposition in $\mathbb{H}.$ Since a non-Hermitian observable $X$ represented by a para-Hermitian operator with discrete spectrum has a complete decomposition, the evolution of a non-Hermitian quantum system by the Heisenberg equation
\begin{equation}\label{eq:HeisenbergEqu}
\mathrm{i} \frac{d X}{d t} = [X, H]
\end{equation}
for the observable $X,$ gives rise to a curve in  $\tilde{\mathcal{W}} (\mathbb{H}).$ This is the reason why $\tilde{\mathcal{W}} (\mathbb{H})$ can be regarded as the observable space, whose geometry induces a geometric structure for a non-Hermitian quantum system.

We equip $\tilde{\mathcal{W}} (\mathbb{H})$ with the Hausdorff distance $D_{\tilde{\mathcal{W}}}$ defined by
\begin{equation}\label{eq:HausdDist}
D_{\tilde{\mathcal{W}}} ( O, O') = \max_{a \in O} \inf_{b \in O'} \| a - b \| + \max_{a \in O'} \inf_{b \in O} \| a - b \|,\quad \forall  O, O' \in \tilde{\mathcal{W}}(\mathbb{H}).
\end{equation}
Then $\tilde{\mathcal{W}} (\mathbb{H})$ is a complete metric space under the distance $D_{\tilde{\mathcal{W}}}.$ Also, we define $\tilde{\mathcal{X}} (\mathbb{H})$ to be the set of all ordered sequences $(|n\rangle \langle n^* |)_{n \ge 1},$ where $\{|n\rangle \langle n^* |:\; n \ge 1\}$'s are all complete decompositions in $\mathbb{H}.$ We equip $\tilde{\mathcal{X}} (\mathbb{H})$ with the distance $D_{\tilde{\mathcal{X}}}$ defined as follows: For $(|n\rangle \langle n^* |)_{n \ge 1}, (|\bar{n}\rangle \langle \bar{n}^* |)_{n \ge 1} \in \tilde{\mathcal{X}} (\mathbb{H}),$
$$
D_{\tilde{\mathcal{X}}} ((|n\rangle \langle n^* |)_{n \ge 1}, (|\bar{n}\rangle \langle \bar{n}^* |)_{n \ge 1}) = \max_{n \ge 1} \| |n\rangle \langle n^* | - |\bar{n}\rangle \langle \bar{n}^* | \|.
$$
Then $\tilde{\mathcal{X}} (\mathbb{H})$ is a complete metric space under $D_{\tilde{\mathcal{X}}}$ such that
$$
\tilde{\mathcal{W}} (\mathbb{H}) \cong \frac{\tilde{\mathcal{X}} (\mathbb{H})}{\Pi (d)},
$$
where $\Pi (d)$ denotes the permutation group of $d$ objects ($d$ denotes the dimension of $\mathbb{H}$), which has a representation in $\mathbb{H}$ as follows: For a given unconditional basis $\{|e_n\rangle \}_{n \ge 1}$ of $\mathbb{H},$
\begin{equation}\label{eq:PermutationGroupRepresentation}
\Pi (d) = \bigg \{ V_\sigma = \sum_{n \ge 1} |e_{\sigma (n)} \rangle \langle e_n^* | \in \mathcal{T} (\mathbb{H}):\; \forall \sigma \in \Pi (d) \bigg \}.
\end{equation}

We denote by
\be
\mathcal{G}_\8 (\mathbb{C}_*) = \big \{ (c_n)_{n \ge 1} \in \mathbb{C}_*^d:\; 0< \inf_{n \ge 1} |c_n| \le \sup_{n \ge 1} |c_n|< \8 \big \}.
\ee
Then $\mathcal{G}_\8 (\mathbb{C}_*)$ is an abelian topological group under pointwise multiplication and has a representation in $\mathbb{H}$ as follows: For a given unconditional basis $\{|e_n\rangle \}_{n \ge 1}$ of $\mathbb{H},$
\begin{equation}\label{eq:ScalarGroupRepresentation}
\mathcal{G}_\8 (\mathbb{C}_*) = \bigg \{ V_{(c_n)} = \sum_{n \ge 1} c_n |e_n \rangle \langle e_n^* | \in \mathcal{T} (\mathbb{H}):\; \forall (c_n)_{n \ge 1} \in \mathcal{G}_\8 (\mathbb{C}_*) \bigg \}.
\end{equation}

\begin{proposition}\label{prop:TopoSpaceQ-system}\rm
For a given unconditional basis $\{|e_n\rangle \}_{n \ge 1}$ of $\mathbb{H},$
$$
\tilde{\mathcal{W}} (\mathbb{H}) \cong \{ \mathcal{G} (V):\; V \in \mathcal{T} (\mathbb{H})\}
$$
with
\begin{equation}\label{eq:FiberForm}
\mathcal{G} (V) = \Big \{ \sum_{n \ge 1} c_n |\sigma (n) \rangle \langle e_n^* |:\; \forall \sigma \in \Pi (d), \forall (c_n)_{n \ge 1} \in \mathcal{G}_\8 (\mathbb{C}_*) \Big \},
\end{equation}
where $|n\rangle = V |e_n\rangle$ for any $n \ge 1,$ and the distance between two elements is defined by
$$
d (\mathcal{G} (V), \mathcal{G} (V')) = \inf \{ \| K - G \|: K \in \mathcal{G} (V), G \in \mathcal{G} (V') \}.
$$
\end{proposition}

\begin{proof}
We need to prove that
$$
\mathcal{X} (\mathbb{H}) \cong \frac{\mathcal{T} (\mathbb{H})}{\mathcal{G}_\8 (\mathbb{C}_*)},
$$
from which we conclude the result.

Indeed, for a fixed unconditional basis $\{|e_n\rangle \}_{n \ge 1}$ of $\mathbb{H},$ we have that $\frac{\mathcal{T} (\mathbb{H})}{\mathcal{G}_\8 (\mathbb{C}_*)} = \{ [V]:\; V \in \mathcal{T} (\mathbb{H}) \}$ with
$$
[V] = V \cdot \mathbb{C}_*^d = \Big \{ \sum_{n \ge 1} c_n |n \rangle \langle e_n^* |:\; \forall (c_n)_{n \ge 1} \in \mathcal{G}_\8 (\mathbb{C}_*) \Big \},
$$
where $|n\rangle = V |e_n\rangle$ for $n \ge 1.$ Define $T: \tilde{\mathcal{X}} (\mathbb{H}) \mapsto \frac{\mathcal{T} (\mathbb{H})}{\mathcal{G}_\8 (\mathbb{C}_*)}$ by
$$
T [(|n\rangle \langle n^* |)_{n \ge 1}] \longmapsto [V]
$$
for any $(|n\rangle \langle n^* |)_{n \ge 1} \in \tilde{\mathcal{X}} (\mathbb{H}),$ where $V$ is the invertible operator so that $|n\rangle = V |e_n\rangle$ for $n \ge 1.$ Then, $T$ is surjective and isometric, and so the required assertion follows. This completes the proof.
\end{proof}

\subsection{Fibre bundles over the non-Hermitian observable space}\label{FibleBundleObNH}

According to \cite{Isham1999}, a bundle is a triple $(E, \pi, B),$ where $E$ and $B$ are two Hausdorff topological spaces, and $\pi: E \mapsto B$ is a continuous map which is always assumed to be surjective. The space $E$ is called the total space, the space $B$ is called the base space, and the map $\pi$ is called the projection of the bundle. For each $b \in B,$ the set $\pi^{-1} (b)$ is called the fiber of the bundle over $b.$ Given a topological space $F,$ a bundle $(E, \pi, B)$ is called a fiber bundle with the fiber $F$ provided every fiber $\pi^{-1} (b)$ for $b \in B$ is homeomorphic to $F.$ For a topological group $G,$ a bundle $(E, \pi, B)$ is called a $G$-bundle, denoted by $(E, \pi, B, G),$ provided $G$ acts on $E$ from the right preserving the fibers of $E$ such that the map $f$ from the quotient space $E/G$ onto $B$ defined by $f (x G) = \pi (x)$ for $x G \in E/G$ is a homeomorphism, namely
 \[
 \xymatrix{
E \ar[d]_{P_G} \ar[rr]^{id} & & E \ar[d]^\pi  \\
E/G \ar[rr]^{f:\cong} & & B }
 \]
where $P_G$ is the usual projection. A $G$-bundle $(E, \pi, B, G)$ is principal if the action of $G$ on $E$ is free in the sense that $x g = x$ for some $x \in E$ and $g \in G$ implies $g=1,$ and the group $G$ is then called the structure group of the bundle $(E, \pi, B, G)$ (in physical literatures $G$ is also called the gauge group, cf. \cite{BMKNZ2003}). Note that, in a principal $G$-bundle $(E, \pi, B, G),$ every fiber $\pi^{-1} (b)$ for $b \in B$ is homeomorphic to $G$ by the freedom of the $G$-action, hence it is a fiber bundle $(E, \pi, B, G)$ with the fiber $G$ and is simply called a principal fiber bundle with the structure group $G.$

Next, we construct principal fiber bundles over the observable space $\tilde{\mathcal{W}} (\mathbb{H}).$ To this end, fix a point $O_0 = \{ | e_n \rangle \langle e_n^*|: n \ge 1 \}$ in $\tilde{\mathcal{W}} (\mathbb{H}).$ For any $O \in \tilde{\mathcal{W}} (\mathbb{H}),$ we write
$$
\mathcal{F}^O_{O_0} = \{ V \in \mathcal{T} (\mathbb{H}):\; V^{-1} O V = O_0 \},
$$
that is, $V \in \mathcal{F}^O_{O_0}$ if and only if $\{ V| e_n \rangle: n \ge 1 \}$ is an unconditional basis such that $O = \{ V| e_n \rangle \langle e_n^*| V^{-1}: n \ge 1 \}.$ Indeed, if $O = \{ |n \rangle \langle n^*|: n \ge 1 \},$ then
$$
\mathcal{F}^O_{O_0} = \mathcal{G} (V) = \bigg \{ \sum_{n \ge 1} c_n | \sigma (n) \rangle \langle e_n^* |:\; \forall \sigma \in \Pi (d), \forall (c_n)_{n \ge 1} \in \mathcal{G}_\8 (\mathbb{C}_*) \bigg \},
$$
where $V$ is an invertible operator so that $|n\rangle = V |e_n\rangle$ for $n \ge 1.$ Also, define
\begin{equation}\label{eq:GaugeGroup}
\mathcal{G}_{O_0} = \bigg \{ \sum_{n \ge 1} c_n |e_{\sigma (n)} \rangle \langle e_n^* |:\; \forall \sigma \in \Pi (d), \forall (c_n)_{n \ge 1} \in \mathcal{G}_\8 (\mathbb{C}_*) \bigg \}.
\end{equation}
By \eqref{eq:PermutationGroupRepresentation} and \eqref{eq:ScalarGroupRepresentation}, $\mathcal{G}_{O_0}$ is a (non-abelian) subgroup of $\mathcal{T} (\mathbb{H})$ generated by $\mathcal{G}_\8 (\mathbb{C}_*)$ and $\Pi (d).$

The (right) action of $\mathcal{G}_{O_0}$ on $\mathcal{F}^O_{O_0}$ is defined as: For any $G \in \mathcal{G}_{O_0},$
$$
(G, V) \mapsto V G
$$
for all $V \in \mathcal{F}^O_{O_0}.$ Evidently, this action is free and invariant, namely $\mathcal{F}^O_{O_0}\cdot G = \mathcal{F}^O_{O_0}$ for any $G \in \mathcal{G}_{O_0}$ and every $O \in \tilde{\mathcal{W}} (\mathbb{H}).$ Note that
$$
\mathcal{T} (\mathbb{H}) = \bigcup_{O \in \tilde{\mathcal{W}} (\mathbb{H})} \mathcal{F}^O_{O_0},
$$
and $\mathcal{F}^O_{O_0}$ is homeomorphic to $\mathcal{G}_{O_0}$ as topological spaces since $\mathcal{F}^O_{O_0} = \mathcal{G} [V]$ for some $V \in \mathcal{T} (\mathbb{H})$ such that $O = \{ V| e_n \rangle \langle e_n^*| V^{-1}: n \ge 1 \}.$

The following is then principal fiber bundles over the observable space.

\begin{definition}\label{df:PrincipalFiber}
Given $O_0 \in \tilde{\mathcal{W}} (\mathbb{H}),$ a principal fiber bundle over $\tilde{\mathcal{W}} (\mathbb{H})$ associated with $O_0$ is defined to be
$$
\xi_{O_0} (\mathbb{H}) = (\mathcal{T} (\mathbb{H}), \Pi_{O_0}, \tilde{\mathcal{W}} (\mathbb{H}), \mathcal{G}_{O_0}),
$$
where $\mathcal{T} (\mathbb{H})$ is the total space, and the bundle projection $\Pi_{O_0}: \mathcal{T} (\mathbb{H}) \mapsto \tilde{\mathcal{W}} (\mathbb{H})$ is defined by
$$
\Pi_{O_0} (V) = O
$$
provided $V \in \mathcal{F}^O_{O_0}$ for (unique) $O \in \tilde{\mathcal{W}} (\mathbb{H}),$ namely $\Pi^{-1} (O) = \mathcal{F}^O_{O_0}$ for every $O \in \tilde{\mathcal{W}} (\mathbb{H}).$

We simply denote this bundle by $\xi_{O_0} = \xi_{O_0} (\mathbb{H}).$
\end{definition}

\begin{remark}\rm
In the sequel, we will see that the fixed point $O_0 \in \tilde{\mathcal{W}} (\mathbb{H})$ physically plays a role of measurement. On the other hand, the point $O_0$ induces a differential structure over the base space $\tilde{\mathcal{W}} (\mathbb{H})$ and determines the geometric structure of $\xi_{O_0},$ namely quantum connection and parallel transportation.
\end{remark}

For any two points $O_0, \bar{O}_0 \in \tilde{\mathcal{W}} (\mathbb{H})$ with $O_0 = \{ | e_n \rangle \langle e_n^*|: n \ge 1 \}$ and $\bar{O}_0 = \{ | \bar{e}_n \rangle \langle \bar{e}_n^* |: n \ge 1 \},$ we define an invertible operator $V_0$ by $V_0 |e_n \rangle = | \bar{e}_n \rangle$ for $n \ge 1.$ Then the map $T: \xi_{O_0} \mapsto \xi_{O'_0}$ defined by $T V = V V^{-1}_0$ for all $V \in \mathcal{T} (\mathbb{H})$ is an isometric isomorphism on $\mathcal{T} (\mathbb{H})$ such that $T$ maps the fibers of $\xi_{O_0}$ onto the fibers of $\xi_{O'_0}$ over the same points in the base space $\tilde{\mathcal{W}} (\mathbb{H}),$ namely the following diagram is commutative:
 \[
 \xymatrix{
\mathcal{T} (\mathbb{H}) \ar[dr]_{\Pi_{O_0}} \ar[rr]^{T} & & \mathcal{T} (\mathbb{H}) \ar[dl]^{\Pi_{O'_0}}  \\
      &  \tilde{\mathcal{W}} (\mathbb{H}) &            }
 \]
that is, $\Pi_{O_0} = \Pi_{O'_0} \circ T.$ Thus, $\xi_{O_0}$ and $\xi_{O'_0}$ are isomorphic as principal fiber bundles (cf. \cite{Isham1999}).

\subsection{Quantum connection}\label{q-connection}

In order to define the suitable concepts of quantum connection and parallel transportation over the principal fiber bundle $\xi_{O_0},$ we need to introduce a differential structure over $\tilde{\mathcal{W}} (\mathbb{H})$ associated with each fixed $O_0 \in \tilde{\mathcal{W}} (\mathbb{H}).$ Indeed, we will introduce a geometric structure over $\xi_{O_0}$ in a certain operator-theoretic sense (cf. \cite{Chen2020}).

Let us begin with the definition of tangent vectors for $\mathcal{G}_{O_0}$ in the operator-theoretic sense. We denote $\mathcal{Q} (\mathbb{H})$ to be the set of all densely defined closed operators in $\mathbb{H}.$

\begin{definition}\label{df:q-tangvectorUgroup}
Fix $O_0 \in \tilde{\mathcal{W}} (\mathbb{H}).$ For a given $V \in \mathcal{G}_{O_0},$ an operator $Q \in \mathcal{Q} (\mathbb{H})$ is called a tangent vector at $V$ for $\mathcal{G}_{O_0},$ if there is a curve $\chi: (-\varepsilon, \varepsilon) \ni t \mapsto V(t) \in \mathcal{G}_{O_0}$ with $\chi (0) = V$ such that for every $h \in \mathcal{D} (Q),$ the limit
$$
\lim_{t \to 0} \frac{V(t) (h) - V (h)}{t} = Q (h)
$$
in $\mathbb{H},$ denoted by $Q = \frac{d \chi (t)}{d t} \big |_{t=0}.$ The set of all tangent vectors at $V$ is denoted by $T_V \mathcal{G}_{O_0},$ and $T \mathcal{G}_{O_0}= \bigcup_{V \in \mathcal{G}_{O_0}} T_V \mathcal{G}_{O_0}.$ In particular, we denote $\mathrm{g}_{O_0} = T_V \mathcal{G}_{O_0}$ if $V =I.$
\end{definition}

Note that given $V \in \mathcal{G}_{O_0}$ with the form $V =\sum_{n \ge 1} c_n |e_{\sigma (n)} \rangle \langle e_n^* |$ for some $\sigma \in \Pi (d)$ and $(c_n)_{n \ge 1} \in \mathcal{G}_\8 (\mathbb{C}_*),$ for every $Q \in T_V \mathcal{G}_{O_0}$ there exists a unique sequence of complex number $(\alpha_n)_{n \ge 1}$ such that
\begin{equation}\label{eq:VertVectStrucGroupExpress}
Q = \sum_{n \ge 1} \alpha_n |e_{\sigma (n)} \rangle \langle e_n^* |.
\end{equation}
In particular, each element $Q \in \mathrm{g}_{O_0}$ is of form
\begin{equation}\label{eq:VertVectLieAlg}
Q = \sum_{n \ge 1} \alpha_n |e_n \rangle \langle e_n^* |,
\end{equation}
where $(\alpha_n)_{n \ge 1}$ is a sequence of complex number. Thus, $T_V \mathcal{G}_{O_0}$ is a linear subspace of $\mathcal{Q} (\mathbb{H}).$

The following is the tangent space for the base space $\tilde{\mathcal{W}} (\mathbb{H})$ in the operator-theoretic sense.

\begin{definition}\label{df:q-tangvectorBaseSpace}

\begin{enumerate}[{\rm 1)}]

\item Fix $O_0 \in \tilde{\mathcal{W}} (\mathbb{H}).$ A continuous curve $\chi: [a, b] \ni t \mapsto O(t) \in \tilde{\mathcal{W}} (\mathbb{H})$ is said to be differential at a fixed $t_0 \in (a, b)$ relative to $O_0,$ if there is a nonempty subset $\mathcal{A}$ of $\mathcal{Q} (\mathbb{H})$ satisfying that for any $Q \in \mathcal{A}$ there exist $\varepsilon>0$ such that $(t_0 -\varepsilon, t_0 + \varepsilon) \subset [a,b]$ and a strongly continuous curve $\gamma: (t_0 -\varepsilon, t_0 + \varepsilon) \ni t \mapsto V_t \in \mathcal{F}^{O(t)}_{O_0}$ such that the limit
$$
\lim_{t \to t_0} \frac{V_t (h) - V_{t_0} (h)}{t - t_0} = Q (h)
$$
for any $h \in \mathcal{D} (Q).$ In this case, $\mathcal{A}$ is called a tangent vector of $\chi$ at $t=t_0$ and denoted by
$$
\mathcal{A} = \frac{d O(t)}{d t}\big |_{t = t_0} = \frac{d \chi (t)}{d t} \big |_{t = t_0}.
$$
We can define the left (or, right) tangent vector of $\chi$ at $t = a$ (or, $t =b$) in the usual way.

\item Fix $O_0 \in \tilde{\mathcal{W}} (\mathbb{H}).$ Given $O \in \tilde{\mathcal{W}} (\mathbb{H}),$ a tangent vector of $\tilde{\mathcal{W}} (\mathbb{H})$ at $O$ relative to $O_0$ is define to be a nonempty subset $\mathcal{A}$ of $\mathcal{Q} (\mathbb{H}),$ provided $\mathcal{A}$ is a tangent vector of some continuous curve $\chi$ at $t=0,$ where $\chi: (-\varepsilon, \varepsilon) \ni t \mapsto O(t) \in \tilde{\mathcal{W}} (\mathbb{H})$ with $\chi (0) = O,$ i.e., $\mathcal{A} = \frac{d O (t)}{d t} \big |_{t=0}.$ We denote by $T_O \tilde{\mathcal{W}} (\mathbb{H})$ the set of all tangent vectors at $O,$ and write $T \tilde{\mathcal{W}} (\mathbb{H}) = \bigcup_{O \in \tilde{\mathcal{W}} (\mathbb{H})} T_O \tilde{\mathcal{W}} (\mathbb{H}).$

\end{enumerate}
\end{definition}

Note that, the tangent vectors for the base space $\tilde{\mathcal{W}} (\mathbb{H})$ is dependent on the choice of a measurement point $O_0.$ This is the same for the total space $\mathcal{T} (\mathbb{H})$ as follows.

\begin{definition}\label{df:q-tangvectorFiberSpace}

\begin{enumerate}[{\rm 1)}]

\item Fix $O_0 = \{ | e_n \rangle \langle e_n^*|: n \ge 1 \}$ in $\tilde{\mathcal{W}} (\mathbb{H}).$ A strongly continuous curve $\gamma: [a, b] \ni t \mapsto T(t) \in \mathcal{T} (\mathbb{H})$ is said to be differential at a fixed $t_0 \in (a, b)$ relative to $O_0,$ if there is an operator $Q \in \mathcal{Q} (\mathbb{H})$ such that $\{e_n\}_{n \ge 1} \subset \mathcal{D} (Q)$ and the limit
$$
\lim_{t \to t_0} \frac{T(t) (h) - T (t_0) (h)}{t - t_0} = Q (h)
$$
for all $h \in \mathcal{D} (Q).$ In this case, $Q$ is called the tangent vector of $\gamma$ at $t=t_0$ and denoted by
$$
Q = \frac{d \gamma (t)}{d t} \Big |_{t = t_0} = \frac{d T (t)}{d t} \Big |_{t = t_0}.
$$
We can define the left (or, right) tangent vector of $\gamma$ at $t = a$ (or, $t =b$) in the usual way.

Moreover, $\gamma$ is called a smooth curve relative to $O_0,$ if $\gamma$ is differential at each point $t \in [a, b]$ relative to $O_0,$ and for any $n \ge 1,$ the $\mathbb{H}$-valued function $t \mapsto \frac{d \gamma (t)}{d t} (e_n)$ is continuous in $[a, b].$

\item Fix $O_0 \in \tilde{\mathcal{W}} (\mathbb{H}).$ For a given $P \in \mathcal{T} (\mathbb{H}),$ an operator $Q \in \mathcal{Q} (\mathbb{H})$ is called a tangent vector of $\xi_{O_0}$ at $P,$ if there exists a strongly continuous curve $\gamma: (-\varepsilon, \varepsilon) \ni t \mapsto P_t \in \mathcal{T} (\mathbb{H})$ with $\gamma (0) = P,$ such that $\gamma$ is differential at $t=0$ relative to $O_0,$ and $Q = \frac{d \gamma (t)}{d t} \big |_{t =0}.$ Denote $T_P \xi_{O_0} (\mathbb{H})$ to be the set of all tangent vectors of $\xi_{O_0}$ at $P$ relative to $O_0,$ and write
$$
T \xi_{O_0} (\mathbb{H}) = \bigcup_{P \in \mathcal{T} (\mathbb{H})} T_P \xi_{O_0} (\mathbb{H}).
$$

\item Fix $O_0 \in \tilde{\mathcal{W}} (\mathbb{H}).$ Given $P \in \mathcal{T} (\mathbb{H}),$ a tangent vector $Q \in T_P \xi_{O_0} (\mathbb{H})$ is said to be vertical, if there is a strongly continuous curve $\gamma: (-\varepsilon, \varepsilon)\ni t \mapsto P_t \in \mathcal{F}^{\Pi (P)}_{O_0}$ with $\gamma (0) = P$ such that $\gamma$ is differential at $t=0$ relative to $O_0,$ and $Q = \frac{d \gamma (t)}{d t} \big |_{t =0}.$ We denote $V_P \xi_{O_0} (\mathbb{H})$ to be the set of all vertically tangent vectors at $P.$
\end{enumerate}
\end{definition}

\begin{remark}\rm
Note that for a given $P \in \mathcal{T} (\mathbb{H}),$ every $Q \in V_P \xi_{O_0} (\mathbb{H})$ with $O_0 = \{|e_n \rangle \langle e_n^* |\}_{n \ge 1}$ has the form
\begin{equation}\label{eq:VertTangentVect}
Q = \sum_{n \ge 1} \alpha_n P |e_n \rangle \langle e_n^* |,
\end{equation}
where $(\alpha_n)_{n \ge 1}$ is a sequence of complex number.
\end{remark}

Given $O_0 \in \tilde{\mathcal{W}} (\mathbb{H}),$ for each $G \in \mathcal{G}_{O_0},$ the right action $R_G$ of $\mathcal{G}_{O_0}$ on $\xi_{O_0}$ is defined by
$$
R_G (V) = V G,\quad \forall V \in \mathcal{T} (\mathbb{H}).
$$
This induces a map $(R_G)_*: T_P \xi_{O_0} (\mathbb{H}) \mapsto T_{R_G(P)} \xi_{O_0} (\mathbb{H})$ for each $P \in \mathcal{T} (\mathbb{H})$ such that
$$
(R_G)_* (Q) = Q G,\quad \forall Q \in T_P \xi_{O_0} (\mathbb{H}).
$$
Since $R_G$ preserves the fibers of $\xi_{O_0},$ then $(R_G)_*$ maps $V_P \xi_{O_0} (\mathbb{H})$ into $V_{R_G(P)} \xi_{O_0} (\mathbb{H}).$

Now, we are ready to define the concept of quantum connection over the observable space.

\begin{definition}\label{df:q-connetion}
Fix $O_0 = \{|e_n \rangle \langle e_n^* |\}_{n \ge 1}$ in $\tilde{\mathcal{W}} (\mathbb{H}).$ A connection on the principal fiber bundle $\xi_{O_0}= (\mathcal{T} (\mathbb{H}), \Pi_{O_0}, \tilde{\mathcal{W}} (\mathbb{H}), \mathcal{G}_{O_0})$ is a family of linear operators $\Omega = \{\Omega_P:\; P \in \mathcal{T} (\mathbb{H}) \},$ where $\Omega_P$ is a linear mapping from $T_P \xi_{O_0} (\mathbb{H})$ into $\mathrm{g}_{O_0}$ for $P \in \mathcal{T} (\mathbb{H}),$ satisfying the following conditions:
\begin{enumerate}[{\rm (1)}]

\item For any $P \in \mathcal{T} (\mathbb{H}),$
\begin{equation}\label{eq:q-ConnectionVertTangVect}
\Omega_P (Q) = P^{-1} Q, \quad \forall Q \in V_P \xi_{O_0} (\mathbb{H}).
\end{equation}

\item $\Omega_P$ depends continuously on $P$ in the sense that if $P_k$ converges to $P_0$ in $\mathcal{T} (\mathbb{H})$ in the uniform operator topology, and if $Q_k \in T_{P_k} \xi_{O_0} (\mathbb{H}), Q_0 \in T_{P_0} \xi_{O_0} (\mathbb{H})$ such that $\lim_k Q_k (e_n) = Q_0 (e_n)$ for all $n \ge 1,$ then
\be
\lim_k \Omega_{P_k} (Q_k) (e_n) = \Omega_{P_0} (Q_0) (e_n), \quad \forall n \ge 1.
\ee

\item For any $G \in \mathcal{G}_{O_0}$ and $P \in \mathcal{T} (\mathbb{H}),$
\begin{equation}\label{eq:GaugeTransConnection}
\Omega_{R_G(P)} [(R_G)_* (Q )] = G^{-1} \Omega_P (Q) G,\quad \forall Q \in T_P \xi_{O_0} (\mathbb{H}),
\end{equation}
namely, $\Omega$ transforms according to \eqref{eq:GaugeTransConnection} under the right action of $\mathcal{G}_{O_0}$ on $\xi_{O_0} (\mathbb{H}).$

\end{enumerate}
Such a connection is simply called an $O_0$-connection.
\end{definition}

Next, we present a canonical example of such quantum connections, which plays a crucial role in the expression of non-Hermitian observable-geometric phases.

\begin{example}\label{Ex:CanonicalConnection}\rm
Fix $O_0  = \{|e_n \rangle\langle e_n^*| \}_{n \ge 1}\in \tilde{\mathcal{W}} (\mathbb{H}),$ we define $\check{\Omega} = \{\check{\Omega}_P: P \in \mathcal{T} (\mathbb{H}) \}$ as follows: For each $P \in \mathcal{T} (\mathbb{H}),$ $\check{\Omega}_P : T_P \xi_{O_0} (\mathbb{H}) \mapsto \mathrm{g}_{O_0}$ is given by
\begin{equation}\label{eq:CanonConnction}
\check{\Omega}_P (Q) = P^{-1} \star Q,\quad \forall Q \in T_P \xi_{O_0} (\mathbb{H}),
\end{equation}
where
$$
P^{-1} \star Q = \sum_{n \ge 1}  \langle e_n^* | P^{-1} Q | e_n \rangle |e_n \rangle \langle e_n^*|.
$$

By \eqref{eq:VertTangentVect}, one has $P^{-1} \star Q = P^{-1} Q \in \mathrm{g}_{O_0}$ for any $Q \in V_P \xi_{O_0} (\mathbb{H}),$ namely $\check{\Omega}_P$ satisfies \eqref{eq:q-ConnectionVertTangVect}. The conditions (2) and (3) of Definition \ref{df:q-connetion} are clearly satisfied by $\check{\Omega}.$ Hence, $\check{\Omega}$ is an $O_0$-connection on $\xi_{O_0}.$ In this case, we write $\check{\Omega}_P = P^{-1} \star d P$ for any $P \in \mathcal{T} (\mathbb{H}).$
\end{example}

\subsection{Quantum parallel transportation}\label{q-ParallelTransport}

This section is devoted to the study of quantum parallel transport over the observable space.

\begin{definition}\label{df:q-lift}
Fix a point $O_0 \in \tilde{\mathcal{W}} (\mathbb{H}).$ For a continuous curve $C_W: [a, b] \ni t \longmapsto O (t) \in \tilde{\mathcal{W}} (\mathbb{H}),$ a lift of $C_W$ with respect to $O_0$ is defined to be a continuous curve
$$
C_P: [a, b] \ni t \longmapsto V(t) \in \mathcal{T} (\mathbb{H})
$$
satisfying the condition that $V(t) \in \mathcal{F}^{O(t)}_{O_0}$ for any $t\in [a, b].$
\end{definition}

\begin{remark}\label{rk:q-lift}\rm
Note that, a lift of $C_W$ depends on the choice of the point $O_0;$ for the same curve $C_W,$ lifts are distinct for different points $O_0.$ For this reason, such a lift $C_P$ is called a $O_0$-lift of $C_W.$
\end{remark}

\begin{definition}\label{df:SmoothCurve}
Fix a point $O_0 \in \tilde{\mathcal{W}} (\mathbb{H}).$ A continuous curve $C_W: [a, b] \ni t \longmapsto O (t) \in \tilde{\mathcal{W}} (\mathbb{H})$ is said to be smooth relative to $O_0,$ if it has a $O_0$-lift $C_P: [a, b] \ni t \longmapsto V(t) \in \mathcal{T} (\mathbb{H})$ which is a smooth curve relative to $O_0.$ In this case, $C_P$ is called a smooth $O_0$-lift of $C_W.$
\end{definition}

Note that, if a continuous curve $C_W: [a, b] \ni t \longmapsto O (t) \in \tilde{\mathcal{W}} (\mathbb{H})$ is smooth relative to $O_0,$ then it is differential at every point $t \in [a, b]$ relative to $O_0.$ Indeed, suppose that $C_P: [a, b] \ni t \longmapsto V(t) \in \mathcal{T} (\mathbb{H})$ is a smooth $O_0$-lift of $C_W.$ For each $t \in [a, b],$ we have $\frac{d C_P (t)}{d t} \in \frac{d O(t)}{d t},$ namely $\frac{d O(t)}{d t}$ is a nonempty subset of $\mathcal{Q} (\mathbb{H}),$ and hence $C_W$ is differential at $t$ relative to $O_0.$

\begin{definition}\label{df:q-HorizontalLift}
Fix $O_0 \in \tilde{\mathcal{W}} (\mathbb{H})$ and suppose $\Omega$ be an $O_0$-connection on $\xi_{O_0} (\mathbb{H}).$ Let $C_W: [0, T] \ni t \longmapsto O (t) \in \tilde{\mathcal{W}} (\mathbb{H})$ be a smooth curve. If $C_P: [0, T] \ni t \longmapsto \tilde{V} (t) \in \mathcal{T} (\mathbb{H})$ is a smooth $O_0$-lift of $C_W$ such that
\begin{equation}\label{eq:ParallelTransfConnectionCond}
\Omega_{\tilde{V} (t)} \Big [ \frac{d \tilde{V} (t)}{d t}\Big ] =0
\end{equation}
for every $t \in [0, T],$ then $C_P$ is called a horizontal $O_0$-lift of $C_W$ with respect to $\Omega.$

In this case, the curve $C_P: t \mapsto \tilde{V} (t)$ is also called the parallel transportation along $C_W$ with the starting point $C_P (0) = \tilde{V}(0)$ with respect to the connection $\Omega$ on $\xi_{O_0} (\mathbb{H}).$
\end{definition}

The following proposition shows the existence of the horizontal lifts in the case of finite dimension.

\begin{proposition}\label{prop:ParallelTransf}\rm
Let $\mathbb{H}$ be a finite-dimensional Hilbert space. Fix $O_0 \in \tilde{\mathcal{W}} (\mathbb{H})$ and let $\Omega$ be an $O_0$-connection on $\xi_{O_0} (\mathbb{H}).$ If $C_W: [0, T] \ni t \longmapsto O (t) \in \tilde{\mathcal{W}} (\mathbb{H})$ is a smooth curve, then for any $V_0 \in \mathcal{F}^{O(0)}_{O_0},$ there exists a unique horizontal $O_0$-lift $\tilde{C}_P$ of $C_W$ with respect to $\Omega$ such that $\tilde{C}_P (0) = V_0.$
\end{proposition}

\begin{proof}
Let $\Gamma: [0, T] \ni t \longmapsto V(t) \in \mathcal{T} (\mathbb{H})$ be a smooth $O_0$-lift of $C_W$ with respect to $\Omega$ with $\Gamma (0) = V_0.$ Note that if $\mathbb{H}$ is a Hilbert space of finite dimension, the condition (2) of Definition \ref{df:q-connetion} implies that the function $t \mapsto \Omega_{\Gamma (t)} \big [ \frac{d \Gamma (t)}{d t} \big ]$ is continuous in $[0, T].$ Then,
\begin{equation}\label{eq:GeodesicEquaGaugeTransf}
\frac{d G (t)}{d t} = - \Omega_{\Gamma (t)} \Big [ \frac{d \Gamma (t)}{d t} \Big ] \cdot G (t)
\end{equation}
with $G (0) = I$ has the unique solution in $[0, T].$ Therefore, $\tilde{C}_P (t) = \Gamma (t) \cdot G (t)$ is the required horizontal $O_0$-lift of $C_W$ for the initial point $V_0 \in \mathcal{F}^{O(0)}_{O_0}.$

To prove the uniqueness, suppose $\check{C}_P: [0, T] \ni t \longmapsto \check{V} (t) \in\mathcal{T} (\mathbb{H})$ be another horizontal $O_0$-lift of $C_W$ for the initial point $V \in \mathcal{F}^{O(0)}_{O_0}.$ Then, for every $t \in [0, T]$ there exists a unique $\check{G} (t) \in \mathcal{G}_0$ such that $\check{C}_P (t) = \tilde{C}_P (t) \cdot \check{G} (t)$ and $\check{G}(0) = I.$ Since
$$
0 = \Omega_{\check{V}(t)} \Big [ \frac{d \check{V}(t)}{d t} \Big ] = \check{G}(t)^{-1} \frac{d \check{G}(t)}{d t},
$$
this follows that $\check{G}(t) = I$ for all $t \in [0, T].$ Hence, the horizontal $O_0$-lift of $C_W$ is unique for the initial point $U \in \mathcal{F}^{O(0)}_{O_0}.$
\end{proof}

\begin{example}\label{Ex:QuantumParallelTransport}\rm
Let $C_P: [0,T] \ni t \mapsto V(t) \in \mathcal{T} (\mathbb{H})$ be a time evolution satisfying the Schr\"{o}dinger equation
\begin{equation}\label{eq:SchrodingerEquTimeUnitaryEvolution}
\mathrm{i} \frac{d V(t)}{d t} = h(t) V (t)
\end{equation}
where $h(t)$'s are time-dependent para-Hermitian operators in $\mathbb{H}.$ Given a fixed point $O_0 = \{|e_n \rangle\langle e_n^*|\}_{n \ge 1}$ in $\tilde{\mathcal{W}} (\mathbb{H}),$ define $C_W: [0, T] \ni t \longmapsto O (t) \in \tilde{\mathcal{W}} (\mathbb{H})$ by $O(t) = V(t) O_0 V^{-1}(t)$ for all $t \in [0, T].$ We define $\tilde{C}_P: [0,T] \ni t \mapsto \tilde{V}(t) \in \tilde{\mathcal{U}} (\mathbb{H})$ by
$$
\tilde{V} (t) = \sum_{n \ge 1} \exp \Big ( - \int^t_0\langle e_n^* | \Big [ V^{-1}(s) \frac{V(s)}{d s} \Big ] | e_n \rangle d s \Big )V(t) | e_n \rangle \langle e_n^*|
$$
for every $t \in [0, T],$ along with the initial point $\tilde{V}(0) = V(0) \in \mathcal{F}^{O(0)}_{O_0}.$ Then $\tilde{C}_P$ is a smooth $O_0$-lift of $C_W$ such that
$$
\check{\Omega}_{\tilde{V} (t)} \Big [ \frac{d \tilde{V}(t)}{d t} \Big ] =0
$$
for all $t \in [0, T],$ where $\check{\Omega}$ is the canonical $O_0$-connection introduced in Example \ref{Ex:CanonicalConnection}. Thus, $\tilde{C}_P$ is the horizontal $O_0$-lift of $C_W$ with respect to $\check{\Omega},$ namely $\tilde{C}_p$ is the parallel transportation along $C_W$ with the starting point $C_P (0) = U(0)$ with respect to the connection $\check{\Omega}$ on $\xi_{O_0} (\mathbb{H}).$
\end{example}

\

{\it Acknowledgments}\; This work is partially supported by the Natural Science Foundation of China under Grant No.11871468.

\bibliography{apssamp}

\begin{thebibliography}{**}

\bibitem{AA1987} Y. Aharonov, J. Anandan, Phase change during a cyclic quantum evolution, {\it Physical Review Letters} {\bf 58} (1987), 1593-1596.

\bibitem{Anandan1988} J. Anandan, Non-adiabatic non-Abelian geometric phase, {\it Physical Letters A} {\bf 133} (1988), 171-175.

\bibitem{AT2014} J.P. Antoine, C. Trapani, Some remarks on quasi-Hermitian operators, {\it Journal of Mathematical Physics} {\bf 55} (2014), 013503: 1-17.

\bibitem{BB1998} C.M. Bender, S. Boettcher, Real spectra in non-Hermitian Hamiltonians having PT symmetry, {\it Physical Review Letters} {\bf 80} (1998), 5243-5346.

\bibitem{BBM1999} C.M. Bender, S. Boettcher, P. N. Meisinger, PT-symmetric quantum mechanics, {\it Journal of Mathematical Physics} {\bf 40} (1999), 2201-2229.

\bibitem{BBJ2002} C. M. Bender, D. C. Brody, H. F. Jones, Complex extension of quantum mechanics, {\it Physical Review Letters} {\bf 89} (2002), 270401: 1-4.

\bibitem{BBJM2007} C. M. Bender, D. C. Brody, H. F. Jones, B.K. Meister, Faster than Hermitian quantum mechanics, {\it Physical Review Letters} {\bf 98} (2007), 040403: 1-4.

\bibitem{BBM2017} C. M. Bender, D. C. Brody, M.P. M\"{u}ller, Hamiltonian for the zeros of the Riemann zeta function, {\it Physical Review Letters} {\bf 118} (2017), 130201: 1-5.

\bibitem{Berry1984} M. V. Berry, Quantal phase factors accompanying adiabatic changes, {\it Proceedings of the Royal Society of London, Series A} {\bf 392} (1984), 45-57

\bibitem{BMKNZ2003} A. Bohm, A. Mostafazadeh, H. Koizumi, Q. Niu, J. Zwanziger,
{\it The Geometric Phase in Quantum Systems,} Springer-Verlag, Berlin, 2003.

\bibitem{Brody2014} D. C. Brody, Biorthogonal quantum mechanics, {\it Journal of Physics A: Mathematical and Theoretical} {\bf 47} (2014), 035305: 1-21.

\bibitem{Chen2020} Z. Chen, Observable-geometric phases and quantum computation, {\it International Journal of Theoretical Physics} {\bf 59} (2020), 1255-1276.

\bibitem{CZ2012} X. Cui, Y. Zheng, Geometric phases in non-Hermitian quantum mechanics, {\it Physical Review A} {\bf 86} (2012), 064104:1-4.

\bibitem{Dunf1958} N. Dunford, A survey of the theory of spectral operators, {\it Bulletin of the American Mathematical Society} {\bf 64} (1958), 217-274.

\bibitem{DS1963} N. Dunford, J.T. Schwartz, {\it Linear Operators Part II. Spectral Theory on Self-adjoint Operators in Hilbert Space,} Interscience, New York, 1963.

\bibitem{DS1971} N. Dunford, J.T. Schwartz, {\it Linear Operators Part III. Spectral Operators,} Wiley-Interscience, New York, 1971.

\bibitem{Dirac1958} P. A. M. Dirac, {\it The Principles of Quantum Mechanics} (Fourth Edition), Oxford University Press, London, 1958.

\bibitem{GW1988} J.C. Garrison, E.M. Wright, Complex geometrical phases for dissipative systems, {\it Physics Letters A} {\bf 128} (1988), 177-181.

\bibitem{Isham1999} C. J. Isham, {\it Morden Differential Geometry for Physicists} (Second Edition), World Scientific, Singapore, 1999.


\bibitem{LT1977}  J. Lindentrauss, L. Tzafriri, {\it Classical Banach Space {\rm I}: Sequence Spaces,} Springer-Verlag, Berlin, 1977.


\bibitem{MM2008} H. Mehri-Dehnavi, A. Mostafazadeh, Geometric phase for non-Hermitian Hamiltonians and its holonomy interpretation, {\it Journal of Mathematical Physics} {\bf 49} (2008), 082105:1-17.


\bibitem{Mosta2010} A. Mostafazadeh, Pseudo-Hermitian representation of quantum mechanics, {\it International Journal of Geometric Methods in Modern Physics} {\bf 7} (2010), 1191-1306.

\bibitem{Mosta2013} A. Mostafazadeh, Pseudo-Hermitian quantum mechanics with unbounded metric operators, {\it Philosophical Transactions of the Royal Society A} {\bf 371} (2013), 20120050: 1-7.

\bibitem{NZ2009} H. Neidhardt, V.A. Zagrebnov, Linear non-autonomous Cauchy problems and evolution semigroups, {\it Advances in Differential Equations} {\bf 14} (2009), 289-340.

\bibitem{vN1955} J. von Neumann, {\it Mathematical Foundations of Quantum Mechanics,} Princeton University Press, Princeton, 1955.

\bibitem{Ovch1993} P.G. Ovchinnikov, Automorphisms of the poset of skew projections, {\it Journal of Functional Analysis} {\bf 115} (1993), 184-189.

\bibitem{PLC2022} I.L. Paiva, R. Lenny, E. Cohen, Geometric phases and the Sagnac effect: Foundational aspects and sensing applications, {\it Advanced Quantum Technologies} {\bf 5} (2022), 2100121.

\bibitem{Pauli1943} W. Pauli, On Dirac'a new method of field quantization, {\it Reviews of Modern Physics} {\bf 15} (1943), 175-207.

\bibitem{Pazy1983} A. Pazy, {\it Semigroups of Linear Operators and Applications to Partial Differential Equations,} Springer-Verlag, New York, 1983.

\bibitem{RS1980I} M. Reed, B. Simon, {\it Method of Mordern Mathematical Physics,} Vol. I, Academic Press, San Diego, 1980.

\bibitem{RS1980II} M. Reed, B. Simon, {\it Method of Mordern Mathematical Physics,} Vol. II, Academic Press, Cambridge, 1980.

\bibitem{Rudin1991} W. Rudin, {\it Functional Analysis,} Second Edition, The McGraw-Hill Companies, Inc., New York, 1991.

\bibitem{SB1988} J. Samuel, R. Bhandari, General setting for Berry's phase, {\it Physical Review Letters} {\bf 60} (1988), 2339-2342.

\bibitem{Schmid2016} J. Schmid, Well-posedness of non-autonomous linear evolution equations for generators whose commutators are scalar, {\it Journal of Evolution Equations} {\bf 16} (2016), 21-50.

\bibitem{SG2017} J. Schmid, M. Griesemer, Well-posedness of non-autonomous linear evolution equations in uniformly convex spaces, {\it Mathematische Nachrichten} {\bf 290} (2017), 435-441.

\bibitem{SGH1992} F.G. Scholtz, H.B. Geyer, F.J. W. Hahne, Quasi-Hermitian operators in quantum mechanics and the variational principle, {\it Annals of physics} {\bf 213} (1992), 74-101.

\bibitem{Sierra2019} G. Sierra, The Riemann zeros as spectrum and the Riemann hypothesis, {\it Symmetry} {\bf 11} (2019), 494:1-37.

\bibitem{Simon1983} B. Simon, Holonomy, the quantum adiabatic theorem, and Berry's phase, {\it Physical Review Letters} {\bf 51} (1983), 2167-2170.

\bibitem{Sj2015} E. Sj\"{o}qvist, Geometric phases in quantum information, {\it International Journal of Quantum Chemistry} {\bf 115} (2015), 1311-1326


\bibitem{STAHJS2012} E. Sj\"{o}qvist, D.M. Tong, L.M. Andersson, B. Hessmo, M. Johansson, K. Singh, Non-adiabatic holonomic quantum computation, {\it New Journal of Physics} {\bf 14} (2012), 103035: 1-10.

\bibitem{ZW2019} Q. Zhang, B. Wu, Non-Hermitian quantum systems and their geometric phases, {\it Physical Review A} {\bf 99} (2019), 032121:1-7.

\end{thebibliography}

\end{document}